\documentclass[12pt, draftclsnofoot, onecolumn]{IEEEtran}
\usepackage{multicol}

\newcommand{\mb}{\mathbf}
\newcommand{\mc}{\mathcal}
\newcommand{\bm}{\boldsymbol}
\newcommand{\bth}{\boldsymbol{\theta}}
\newcommand{\bphi}{\boldsymbol{\phi}}
\newcommand{\siga}{\sigma_{\mc A}}
\newcommand{\h}{\text{hom}}
\newcommand{\hq}{Q_{\h, \bphi}}
\usepackage{amsfonts,upgreek,pifont,amssymb,enumerate,color,algorithm,theorem}
\usepackage{algpseudocode}
\usepackage{cite,graphicx,epstopdf}
\usepackage{mathtools}
\usepackage{subcaption}
\newcommand{\qedsymbol}{$\blacksquare$}

\newenvironment{proof}
    {
      \emph{Proof.}
    }
    {
      \hfill\qedsymbol
    }
\newtheorem{Lemma}{Lemma}
\newtheorem{Prop}{Proposition}

\newtheorem{Definition}{Definition}
\newtheorem{Remark}{Remark}
\begin{document}
\title{Multi-Agent Reinforcement Learning for Cooperative Coded Caching via Homotopy Optimization}
\author{
Xiongwei~Wu, {\it Student Member, IEEE}, Jun Li, {\it Senior Member, IEEE}, Ming~Xiao, {\it Senior Member, IEEE}, P. C. Ching, {\it  Fellow, IEEE}, H. Vincent Poor, {\it Fellow, IEEE}
\thanks{This work was supported in part by the U.S. National Science Foundation under Grant CCF-1908308  and in part by a Princeton Schmidt Data-X Research Award.}
\thanks{X. Wu and P. C. Ching are with the Department of Electronic Engineering, The Chinese University of Hong Kong, Shatin, Hong Kong SAR of China  (e-mail: xwwu@ee.cuhk.edu.hk; pcching@ee.cuhk.edu.hk). J. Li is with the School of Electronic and Optical Engineering, Nanjing University of Science and Technology, Nanjing 210094, China. (e-mail:  jun.li@njust.edu.cn). M. Xiao is with the School of Electrical Engineering and Computer Science, Royal Institute of Technology (KTH), Stockholm, Sweden (email: mingx@kth.se). H. V. Poor is 
with the Department of Electrical Engineering, Princeton University, Princeton, NJ 08544. (e-mail: poor@princeton.edu).}
}
\maketitle
\vspace*{-4em}

\begin{abstract}
Introducing cooperative coded caching into small cell networks is a promising approach to reducing traffic loads. By encoding content via maximum distance separable (MDS) codes, coded fragments can be collectively cached at small-cell base stations (SBSs) to enhance caching efficiency. However, content popularity is usually time-varying and unknown in practice. As a result, cache contents are anticipated to be intelligently updated by taking into account limited caching storage and interactive impacts among SBSs. In response to these challenges, we propose a multi-agent deep reinforcement learning (DRL) framework to intelligently update cache contents in dynamic environments. With the goal of minimizing long-term expected fronthaul traffic loads, we first model dynamic coded caching as a cooperative multi-agent Markov decision process. Owing to MDS coding, the resulting decision-making falls into a class of constrained reinforcement learning problems with continuous decision variables. To deal with this difficulty, we custom-build a novel DRL algorithm by embedding homotopy optimization into a deep deterministic policy gradient formalism. Next, to empower the caching framework with an effective trade-off between complexity and performance, we propose centralized, partially and fully decentralized caching controls by applying the derived DRL approach. Simulation results demonstrate the superior performance of the proposed multi-agent framework. 
\end{abstract}

\section{Introduction}
In recent years, there has been a significant surge in mobile data traffic,
which is anticipated to impose a heavy traffic burden on wireless networks for the foreseeable future \cite{index2016cisco,bastug2014living}. As a consequence, wireless networks could  become very congested, and thus find it difficult to satisfy user requests with satisfactory quality of service. To cope with this challenge, edge caching has been proposed as a promising solution towards fifth generation (5G) communications and beyond \cite{bastug2014living}. By endowing caching units in wireless edge nodes, e.g., small-cell base stations (SBSs), popular content can be pre-fetched close to users. Subsequently, the caching content is able to be delivered to users without duplicated transmissions in fronthaul and backhaul links \cite{bastug2014living}. This process significantly decreases traffic loads, alleviates network congestion, reduces delay, and thus improves system performance \cite{liu2016caching,li2015distributed,li2016pricing}.

In general, caching policies should be designed according to system features, e.g., users arrivals and content popularity, to better satisfy user demands. These features, in practice, usually exhibit unknown and temporal dynamics. For instance, content popularity is generally time-varying because the most popular content at the current epoch may not receive the highest attention in the future; and mobile users could change locations as time passes \cite{bettstetter2001mobility}. Thus, with a limited caching storage, it is crucial to learn how to reasonably update cache contents given the observations of system features. Fortunately, by embedding deep learning into reinforcement learning (RL), deep RL (DRL) has emerged as an effective tool to address decision-making in dynamic environments \cite{mnih2015human}. 
This artificial intelligence technique can be leveraged to learn an optimal policy to maximize long-term performance criteria through interactions with environments \cite{mnih2015human, luong2019applications}. 
In this way, utilizing DRL is envisioned to empower ``intelligent'' caching, i.e., updating caching resources by tracking and adapting to dynamic features of wireless networks \cite{sadeghi2019deep}. 

 
\subsection{Related Work}
Prior studies generally investigated the potentials of edge caching by optimizing average performance criteria. For instance, the studies in \cite{blaszczyszyn2015optimal,li2015delay,li2016weighted} examined effective caching strategies to alleviate traffic loads, and to reduce system cost and download latency. With the aid of caching resources at SBSs, the studies in \cite{tao2016content,wu2019latency} investigated the joint design of SBS beamforming and clustering. Caching strategies in these studies were designed to store either entire content items or uncoded fragments, which are referred to as {\it uncoded caching}. 

To further improve caching efficiency, {\it coded caching} has recently gained considerable research attention. The study in \cite{maddah2014fundamental} proposed a novel coded caching scheme, which provides a global caching gain relating to cumulative storage over all caching units. The research in \cite{bioglio2015optimizing,liao2017coding} investigated cooperative coded caching by utilizing maximum separable distance (MDS) codes to reduce traffic loads. 
MDS coded caching was also examined in \cite{wu2020jointLongTermCacheUpdating,wu2019jointMDS} to augment SBS collaboration and thus offers great advantages to lower latency and reduce power consumption compared with {\it uncoded caching}. The above-mentioned studies mainly investigated offline caching policies by assuming time-invariant content popularity distributions. 

To exploit dynamic features in wireless networks, extant works have been devoted to designing caching policy by using RL. The study in \cite{sadeghi2018optimal} utilized Q-learning to find an optimal caching policy to minimize network cost. To counter the curse of dimensionality in conventional RL, DRL-based caching polices were advocated in \cite{he2017deep,sadeghi2019deep,wei2018joint,wu2019dynamic} by using deep neural networks (DNNs) as function approximators. Moreover, the study in \cite{zhong2020deep}  proposed a multi-agent DRL framework to maximize cache hit ratios for centralized and decentralized settings. The authors in \cite{xu2020collaborative} investigated cache placement by using cooperative multi-agent multi-armed bandit learning at small cell networks (SCN). A decentralized caching scheme was proposed in \cite{Wang2020federated} by utilizing federated deep reinforcement learning. 
Nevertheless, the research in \cite{sadeghi2018optimal,wu2019dynamic,he2017deep, sadeghi2019deep, wei2018joint, zhong2020deep, xu2020collaborative,Wang2020federated} focused on {\it uncoded caching}. That is, each content item is entirely cached without exploring SBSs cooperatively fetching coded fragments of each content item. 

\subsection{Contributions}
Indeed, cooperatively pre-fetching MDS coded fragments has been proven to significantly alleviate traffic loads and thus reduce latency, as well as transmission cost over storing uncoded fragments at SCN \cite{liao2017coding,wu2019jointMDS,wu2020jointLongTermCacheUpdating}. 
Specifically, since one can distribute MDS coded fragments of a content item to multiple SBSs, mobile users can access a cluster of SBSs simultaneously to download coded fragments of the desired content item. To date, very few works have investigated how to ``intelligently'' update MDS coded content items under dynamic environments (e.g., time-varying content popularity). Under centralized control,  
a prior study \cite{zhang2019accelerated} utilized deep deterministic policy gradient (DDPG) to explore MDS coded caching under dynamic content popularity with the aid of predicting user requests; and study \cite{gao2020reinforcement} utilized Q-learning with function approximation to investigate coded caching.

It is worth noting that, unlike uncoded caching schemes \cite{luong2019applications, mnih2015human, sadeghi2018optimal,he2017deep, sadeghi2019deep, wei2018joint, wu2019dynamic, zhong2020deep, xu2020collaborative} addressing binary decision-making, coded caching essentially entails continuous caching decisions that are subject to storage constraints. Accordingly, the resulting decision-making for dynamic coded caching is a constrained RL with a continuous action space.
Quantizing actions into discrete values or directly applying conventional DRL algorithms may not be able to efficiently handle this constrained RL problem. In addition, a centralized control could lead to excessive communication overhead because the cloud processor (CP) needs frequent communications with SBSs to aggregate information and inform SBSs of their caching decisions. As the number of SBSs increases, the dimension of continuous states and actions in a centralized control would be very large, and optimal caching is computationally prohibitive. As a consequence, it is very challenging, but essential, to design efficient DRL algorithms for cooperative coded caching. 
 
To bridge the research gap identified above, in this work, we investigate cooperative coded caching design at SCN with temporally evolving content popularity. 
In particular, we address the following fundamental issues for empowering ``intelligent'' caching: 
{\bf i)  how to design efficient DRL algorithms for a constrained RL problem with continuous decision variables; and ii) how to develop a multi-agent DRL-based framework with different levels of controls to obtain an equitable trade-off between performance and complexity.}

The main contributions of this work are summarized as follows:
\begin{enumerate}[$\bullet$]
  \item To the best of our knowledge, this is the first work to investigate a multi-agent DRL framework for MDS coded caching under time-varying content popularity.
  Specifically, we model cache updating for MDS coded caching as a cooperative multi-agent Markov dynamic process (MDP). With the goal of minimizing long-term expected cumulative fronthaul traffic loads, we judiciously define the system state, local observations and action space of each agent, as well as caching reward. Our formulated problem is a continuous RL with action constraints. We also characterize optimal decisions in a closed form.

  \item As a core technical contribution, we reformulate a general constrained RL problem, whose action space is inefficient to be satisfied through designing DNNs, into a tractable form that can be dealt with by utilizing homotopy optimization. 
  Then, we custom-build a novel DRL, i.e., homotopy deep deterministic policy gradient (HDDPG), through recasting the basic elements of RL and unfolding the iterative process of homotopy optimization. The novelty of this approach lies in introducing a reasonable cumulative penalty to the objective of RL, and then properly manipulating it by using homotopy optimization. 

  \item To endow the proposed DRL caching framework with different levels of control, we generalize the proposed HDDPG from centralized control to partially and fully decentralized controls. Specifically, in the centralized control, the CP coordinates SBSs to conduct cache updating by using global information. To reduce complexity and communication overhead, we then propose a partially decentralized control by allowing SBSs to make decisions locally, but their polices are learned in a centralized manner. In the fully decentralized control, each SBS works as an independent learner and trains its caching policy independently based on local observations. The proposed decentralized controls could obtain a desirable trade-off between complexity and performance, and thus have the potential to handle large-scale wireless networks.
\end{enumerate}

The remainder of this paper is organized as follows. 
Sec. II presents the problem statement. Sec. III introduces the proposed DRL. Sec. IV develops a centralized cooperative coded caching design. Sec. V proposes a partially decentralized caching design, and Sec. VI proposes a fully decentralized caching design. Sec. VII presents performance evaluations, and Sec. VIII concludes the paper. 

\section{Problem Statement}
\subsection{MDS Coded Caching at SCN}
As illustrated in Fig. \ref{system}, we consider a SCN, in which a total of $B$ SBSs are densely deployed and thus are capable of cooperatively providing communication services for users. Each SBS is endowed with a cache unit, which can cache popular content from the CP through a capacity-limited fronthaul. The CP is further connected to the core network through a backhaul. Suppose that a catalog of $F$ content items are available in CP. For ease of discussion, all of the content items are of the same size $s$ bits, and each cache unit has a storage of $L\times s$ bits.  Let $\mc B = \{1, \cdots, B\}$ and $\mc F = \{1, \cdots, F\}$ denote the indices of SBSs and content items, respectively.

To reduce traffic loads on the capacity-limited fronthaul and provide better services for mobile users, SBSs can  proactively cache popular content. By applying MDS codes, each content item of size $s$ bits is able to be encoded into a sufficiently long sequence of parity bits, and any $s$ parity bits are sufficient to reconstruct the original content item \cite{wu2020jointLongTermCacheUpdating,liu2013mixed,liao2017coding}. Moreover, in practice, MDS coding can be implemented by Raptor codes only with a very small redundancy \cite{liu2013mixed,shokrollahi2006raptor}. 
Therefore, SBSs with limited caching storage can cooperatively and collectively cache these coded parity bits, so as to satisfy user requests locally as much as possible. More precisely, we define the cache allocation matrix as $\bm L = [l_{f,b}] \in \mathbb{R}^{F\times B}$, where element $l_{f,b} \in [0,1]$ denotes the proportion of parity bits encoding content item $f$ that are stored at SBS $b$, $\forall b \in \mc B, f\in \mc F$. 
Owing to the storage limit at SBSs, cache allocation needs to satisfy 
$\sum_{f \in \mc F} l_{f,b} \leq L, \forall b$. 
It is worth noting that, the parity bits of a content item available at SBS $b$ should be independent of that cached at other SBSs; thus, users can always download distinct coded content items from multiple SBSs \cite{liao2017coding}. 
To guarantee this caching diversity among SBSs, the encoded information sequence of parity bits of every content item should be sufficiently long, e.g., larger than $Bs$. 
\begin{figure}[h]
  \centering
  \includegraphics[scale=0.6]{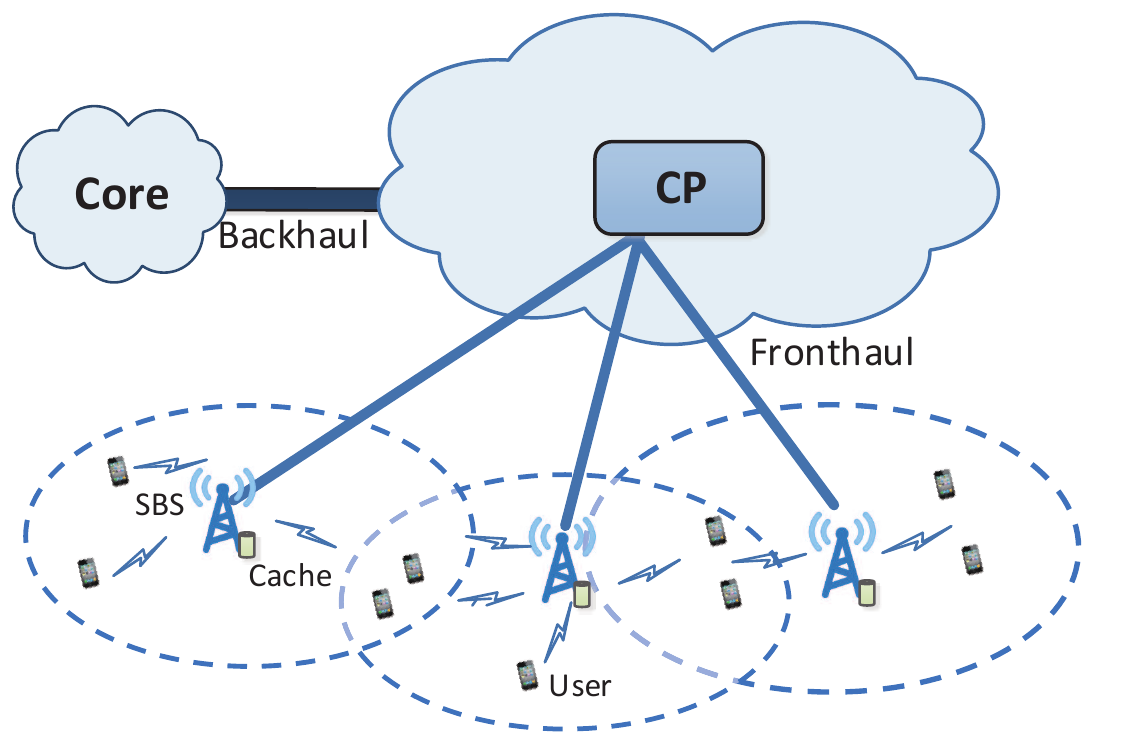}
  \caption{A downlink cache-enabled SCN. }
  \label{system}
\end{figure}

In what follows, we introduce how SBSs cooperatively transmit coded content items to mobile users. Specifically, the operation cycle of a SCN is slotted into a series of epochs, indexed by $t = 0, 1, \cdots$. 
For each epoch $t$, a number of active users $\mc K^t$ are randomly distributed in the horizontal plane. We also assume that the duration of each epoch is relatively short, such that active users are considered to be quasi-static during a single epoch. Each user $k \in \mc K^t$ is able to be served by a local SBS cluster, which is specified by the communication radius \cite{bharath2016learning}. 
Therefore, network connectivity for epoch $t$ is denoted by $\mc E^t = \{e_{k,b}^t\}_{k \in \mc K^t, b \in \mc B}$; each element $e_{f,b}^t = 1$ if $d_{k,b}^t \leq r_0$, where $r_0$ denotes the communication radius of each SBS, and $d_{k,b}^t$ denotes the distance between user $k$ and SBS $b$; otherwise, $e_{f,b}^t = 0$. 
For notational convenience, we collect all active users served by SBS $b$ at epoch $t$ as set $\mc K_b^t = \{ k\in \mc K^t| e_{k,b}^t = 1\}$, for $\forall b \in \mc B$. 
Each user $k \in \mc K^t$ is assumed to request one content item at a single epoch. 
Evidently, when content item $f$ is not fully stored at user $k$ neighboring SBSs, i.e., 
\begin{align}
      \sum_{b\in \mc B} l_{f,b} e_{k,b}^t < 1,
\end{align}
the missing part ($1 - \sum_{b\in \mc B} l_{f,b} e_{f,b}^t) s $ needs to be transmitted by the CP via fronthaul. This event is referred to as \emph{cache miss}, which introduces additional fronthaul traffic loads. We summarize all of the key notations\footnote{Without further definition, in this paper, notation $(\cdot)^t$ denotes the value of $(\cdot)$ taken at epoch $t$.} in Table \ref{table:parameters}.
\begin{table}[h!]
\centering
\caption{Summary of main notations}
\begin{tabular}{c c}
  \hline
      $b, \mc B$  & SBS index, indices of all SBSs  \\
      $k, \mc K_b^t, \mc K^t$ & User index, active users observed by SBS $b$, and all of the active users at epoch $t$\\ 
      $f, f_k^t, \mc F$ & Content item index, index of content item request by user $k$ at epoch $t$, and indices of all content items\\
      $l_{f,b}, \bm L$ & Fraction of content item $f$ stored at SBS $b$, and cache allocation matrix\\
      $L$ & Normalized caching storage at each SBS\\
      $e_{k,b}^t, \mc E^t$ & Whether or not user $k$ can access SBS $b$, and network connectivity at epoch $t$\\
  \hline
\end{tabular}
\label{table:parameters}
\end{table}
As a result, to mitigate traffic burden on fronthaul, optimized caching policies depend highly on knowledge of network connectivity, i.e., $\mc E^t$, and content popularity, i.e., $\{p_{f}, \forall f\}$, where $p_f$ denotes the probability of content item $f$ being requested. These elements generally exhibit unknown and time-varying dynamics in practice. 
Indeed, caching content needs to be temporally updated based on historical observations in order to provide better download services for future requests. We thereby introduce a dynamic cooperative coded caching problem in the following subsection.

\subsection{Cooperative Multi-Agent MDP}
In the cooperative coded caching, SBSs are anticipated to collaboratively cache the coded content items, which can be specified by optimizing continuous variables $\{l_{f,b}^t\}$.  
Therefore, we formulate the considered cooperative coded caching problem as a cooperative multi-agent MDP. 

\begin{Definition}
	A cooperative multi-agent MDP is specified by a tuple $(\mc B, \mc S, \mc A, \mc P, R, \gamma)$, where $\mc B$ denotes the set of agents; $\mc S$ denotes the state space, which aggregates all of the agent local observations $\bm S = \cup_{b \in \mc B} \{\bm S_b\}$; and $\mc A$ denotes the action space for a joint action $\bm A = \cup_{b \in \mc B} \{\bm A_b\}$. Let $\mc A_b$ be the action space of $\bm A_b$, and then $\mc A = \cup_{b \in \mc B} \{ \mc A_b\}$. $\mc P$ collects all of the transition probability Pr$\{\bm S'| \bm S, \bm A\}$ for $\forall \bm S, \bm S' \in \mc S, \bm A \in \mc A$. All of the agents share a common reward $R$ after they cooperatively take actions $\{\bm A_b\}_{b\in \mc B}$. $\gamma \in [0,1)$ denotes a discount factor.
\end{Definition}
As aforementioned, user requests are expected to be satisfied by SBSs locally as much as possible; otherwise, the missing fragments could introduce additional traffic burden and transmission delay on the fronthaul. Therefore, in this paper, our goal is to minimize the expected fronthaul traffic loads. Accordingly, the basic elements in a cooperative multi-agent MDP are defined as follows.  

\subsubsection*{State} We assume that a user request can be observed by his or her neighboring SBSs only.   
Consequently, SBS $b$ has local observation of the environment, which is defined as follows:
\begin{align}
  \bm S_b^t = \left[\{f_k^t\}_{k \in \mc K^t_b}, \{\mc E^t_k\}_{k \in \mc K_b^t}, 
  \{l_{f,b}^t\}_{f\in \mc F} \right], 
\end{align}
where $f_k^t \in \mc F$ denotes the index of the content item requested by user $k$ at epoch $t$; and  
$\mc E_k^t = \{e^t_{k,b'} \}_{ b' \in \mc B}$ implies the strategy of SBS collaboration in order to satisfy user $k$'s request, 
which can be acquired by knowing user location.
By aggregating observations of all SBSs, the system state is defined as: 
\begin{align}
  \bm S^t = \left[ \{f_k^t\}_{k \in \mc K^t}, \mc E^t, \bm L^t\right].
\end{align}

\subsubsection*{Action} By the end of each epoch $t$, all SBSs need to update their cached content. Accordingly, we define the action of SBS $b$ at the current epoch as $\bm A_b^t = [a_{f,b}^t, \forall f\in \mc F] $, where element $a_{f,b}^t = l_{f,b}^{t+1}$; and the corresponding action space is given by: 
\begin{align} 
   \mc A_b = \left\{\bm A_b^t | 0 \leq a_{f,b}^t \leq 1,\forall f\in \mc F, \textstyle \sum_{f\in \mc F} a_{f,b}^t \leq L\right\}. \label{eq:action_space}
\end{align}  
As such, a joint action can be given by $\bm A^t = [a_{f,b}^t, \forall f\in \mc F, b \in \mc B]$.


\subsubsection*{Reward} After executing joint action $\bm A$, the system state turns into $\bm S^{t+1} $ with {\it transition probability} Pr$\{\bm S^{t+1} | \bm S^t, \bm A^t \}$. In this cooperative task, all of the agents shall receive a common reward $R(\bm S^{t+1}, \bm S^t, \bm A^t)$, which indicates how good a joint action $\bm A^t$ is.
Therefore, it should be consistent with the goal of reducing fronthaul traffic loads. 
It is clear that the total traffic loads for updating caching resources and satisfying user requests in the coming epoch are given by:
\begin{align}
  C^{t+1} = \sum_{f \in \mc F, b \in \mc B} \max\left\{l_{f,b}^{t+1} - l_{f,b}^t, 0\right \}s + \sum_{k \in \mc K^{t+1}} \max \left\{ 0, 1 - \sum_{b \in \mc B} l_{f,b}^{t+1} e_{k,b}^{t+1} \bigg|_{f = f_k^{t+1}}\right\}s. \label{eq:cost}
\end{align}
Accordingly, we design the reward as: 
\begin{align}
  R^{t+1} = R(\bm S^{t+1}, \bm S^t, \bm A^t) \triangleq - \frac{ C^{t+1}}{|\mc K^{t+1} |s}, 
\end{align}
which indicates how much traffic loads are imposed in order to satisfy each content request in the coming epoch after performing cache updating.

Toward this end, the goal of this study is to find a cooperative caching policy $\pi^*$, which maximizes the total expected cumulative caching reward, i.e.:
\begin{align} 
  \pi^* = \mathop{\arg\max}\limits_{\pi \in \Pi} \mathbb{E} \left[ V \big| \pi \right],  \label{prob:mamdp}
\end{align}   
and the cumulative reward is defined as:
\begin{align}
   V  \triangleq \sum_{t = 0}^{\infty} (\gamma)^{t} R^{t+1},\label{eq:return} 
\end{align} 
where $\pi$ denotes a mapping from state space to action space; $\Pi$ denotes the set of feasible caching policies; and the expectation is over all of the rewards $\{R(\bm S^{t+1}, \bm S^t, \bm A^t)\}$. Furthermore, a characterization for optimal decisions is presented in the following proposition.
\begin{Prop}
\label{prop:opticon}
  Consider that all SBSs are fully loaded at the initial epoch, i.e., $\sum_{f\in \mc F} l^{0}_{f,b} = L, \forall b \in \mc B$. There exists an optimal decision sequence $\{(\bm A^t)^*\}$ satisfying
$
    \sum_{f \in \mc F} (a_{f,b}^t)^* = L, \forall b \in \mc B, 
$
where $\bm (A^t)^* = [(a_{f,b}^t)^*]$, for any $t \geq 1$. 
\end{Prop}
\begin{proof}
  See Appendix \ref{appen:A}.
\end{proof}
\begin{Remark}
Proposition \ref{prop:opticon} implies that optimal caching decisions are likely to be the case where caching units are fully loaded. 
This result is reasonable and would provide further insight for algorithm design.
Nevertheless, calculating an optimal cooperative caching policy offline depends on knowledge of network dynamics (e.g., transition probability $Pr\{\bm S' | \bm S, \bm A\}$), which is generally difficult to obtain in real applications. Even if this knowledge could be obtained, problem \eqref{prob:mamdp} is still intractable due to its no closed-form expression. 
In view of this, one can resort to DRL to handle this problem through utilizing historical experiences without knowing exact dynamic information. On the other hand, as previously mentioned, the decision variables $\{a_{f,b}^t\}$ for cooperative coded caching are continuous. 
Simply quantizing these variables into discrete values may lead to performance loss as well as an exponentially large number of actions. For example, consider a very small scenario: three SBSs store parity bits of 10 content items, and each continuous decision variable $a_{f,b}^t$ is coarsely quantized into five discrete values within $[0,1]$; then, the resulting actions are $5^{30}$ at most,  leading to value-based RL algorithms (e.g.,  deep Q learning) that are intractable. 
\end{Remark}

\section{A Novel Homotopy DDPG}
To develop a working DRL algorithm for the considered problem, we first introduce a policy-based RL and identify the arising challenges. Then, we recast a general policy based RL problem with constraints into a tractable form, which is suitable to be addressed by leveraging homotopy optimization. Finally, we custom-build a novel DRL algorithm by embedding homotopy optimization into DDPG. 

\subsection{Fundamentals of DDPG}
DDPG is one of the policy-based RL algorithms, which is widely used to handle continuous decision-making \cite{Lillicrap2015ContinuousCW}. Built upon actor-critic architectures, this algorithm employs DNNs as function approximators to learn a deterministic policy that can map high-dimensional states into feasible continuous actions. 
Typically, a DDPG-based RL framework consists of two networks, i.e., 
critic and actor, which are detailed as follows. 
\subsubsection*{Actor} The actor network corresponds to a deterministic policy, which can generate an action $\bm A$ under a given system state $\bm S$, i.e.,
$
\bm A = \pi_{\bth} (\bm S),  
$ 
and $\bth$ is the parameter of the associated DNN. 
This parametrized policy $\pi_{\bth}(\cdot)$ aims to maximize the expected cumulative reward, i.e.:
\begin{align}
  J(\bth) =   \mathbb{E} \left[ V \big| \pi_{\bth} \right].
\end{align}

\subsubsection*{Critic} The critic network $ Q_{\bphi} (\bm S, \bm A) $  serves as an estimator to predict an action-value function (also termed as Q-function), i.e., 
$
 \mathbb{E} \left[ V| \bm S, \bm A, \pi_{\bth} \right],
$
and $\bphi$ denotes the parameter of the associated DNN. 
In general, the critic is designed to fine-tune the actor, which yields
\begin{align}
  \pi_{\bth} (\bm S) = \arg\max_{\bm A \in \mc A} {Q_{\bphi}(\bm S, \bm A)}. 
\end{align}
By recalling \eqref{eq:return}, it is expected to have the following recursive equation:
\begin{align}
  Q_{\bphi}(\bm S, \bm A) = \mathbb{E}_{\bm S', R|\bm S, \bm A} \left [R + \gamma Q_{\bphi}\left(\bm S', \pi_{\bth} (\bm S')\right) \right], 
\end{align}
where $\bm S'$ denotes the subsequent state after taking $\bm A$ under state $\bm S$; $R$ denotes the corresponding instant reward; and the expectation is over all of the possible occurrences of $(\bm S', R)$. 

\subsubsection*{Learning Algorithm} As a category of policy gradient approaches, actor parameter $\bth$ is updated by using stochastic gradient descent, where the gradient of the policy can be given by {\it Deterministic Policy Gradient Theorem} \cite{Lillicrap2015ContinuousCW}; concerning the critic network, parameter $\bphi$ is updated according to {\it Temporal Difference}. Readers are referred to \cite{Lillicrap2015ContinuousCW} for greater details. 

Although DDPG has achieved great success in addressing many continuous decision-making tasks, the action space in our problem (defined by \eqref{eq:action_space}) could restrain it from being efficient. Specifically, to confine the output of the actor to be feasible, a simple idea is to use the activation function {\it SoftMax} to normalize the output of the last hidden layer, which is then filtered by multiplying a scaling factor (e.g., $L$). This idea has been used in \cite{zhang2019accelerated}. We should point out that the resulting elements (i.e., $a_{f,b}^t, \forall f,b$) could surpass 1 when $L \gg 1$; directly clipping it to 1 may lead to a very poor caching decision if an element $a_{f,b}^t = L$ exists. This practice contradicts Proposition \ref{prop:opticon} and could degrade the performance of DDPG. In the following subsections, we formally analyze this issue and propose an efficient approach to overcome this challenge. 

\subsection{A Homotopy Optimization Based Approach}
For a class of RL problems, the corresponding action space $\mc A$ could be some constraints inefficient to be directly satisfied through designing DNNs, i.e., $\mu_{\bth}$. More specifically, we consider the following situation: let set $\mc A_{\bth}$ collect all of the proto-actions as a result of $\mu_{\bth} (\bm S), \forall \bm S \in \mc S$; and feasible actions $\bm A$ may only lie on a subset of $\mc A_{\bth}$, i.e., $\mc A \subseteq \mc A_{\bth}$. To deal with this issue, a straightforward approach is to use a mapping function $\siga(\cdot)$, which can project a proto-action $\mu_{\bth} (\bm S)$ into the action space, i.e., $\sigma_{\mc A} \left[\mu_{\bth} (\bm S) \right] \in \mc A$. Thus, a feasible policy function can be given by $\pi_{\bth}(\cdot) = \siga[\mu_{\bth}(\cdot)]$.
Accordingly, the associated policy-based RL problem is supposed to take the following form:
\begin{align}
  \max_{\bth} J(\bth|\sigma_{\mc A})  \triangleq \mathbb{E} \left[ V \big| \mu_{\bth}, \sigma_{\mc A} \right]. \label{prob:pararl}
\end{align}

Nevertheless, for many constrained RL applications, poor actions are likely to be generated after projection. Like the example in the previous subsection, mapping a proto-action with a dominant element $a_{f,b}^t = L$ into an one-hot vector may lead to a very sparse caching vector; this case implies a very low caching resource available at SBSs. When frequently encountering  this instance during training, using mapping methods may not guarantee network parameters to be efficiently updated. Thus, it will lead to a suboptimal policy.

To remedy this method, a natural idea is to seek a proper way to penalize the performance loss caused by mapping a proto-action $\mu (\bm S^t)$ into a feasible one given any state $\bm S^t$. 
 Let $g(\cdot| \siga)$ be a general penalty function, which needs to be designed according to the corresponding problem. In the coded caching problem, inspired by Proposition \ref{prop:opticon}, a penalty function can be given by
\begin{align} 
 g(\bm S|\siga) = BL - \|\text{vec}(\siga(\bm A)\|_1\big|_{\bm A = \mu_{\bth} (\bm S)},       
\end{align}  
where $\|\cdot\|_1$ denotes $l_1$-norm; and this penalty indicates the remaining storage over all caching units after taking action $\siga(\bm A)$.

The proposed approach is then built upon maximizing a homotopy function:
\begin{align}
   &\max_{\bth} J_{\h}(\bth| \lambda, \sigma_{\mc A})  \triangleq J(\bth|\sigma_{\mc A}) + \lambda  G,  \label{prob:hrl}
\end{align}
where the discount cumulative penalty, i.e., $ G = \mathbb{E} \left[ \sum_{t = 0}^{+\infty} (\gamma)^{t} g (\bm S^{t}|\siga) \right]$, 
is finite due to a discount factor $\gamma \in [0,1)$; and $\lambda \leq 0$ is a homotopy variable, such that:
\begin{align}
  J_{\h}(\bth|\lambda, \sigma_{\mc A}) =
  \begin{cases}
    J(\bth|\sigma_{\mc A}), & \lambda = 0, \\
    J_{\h}(\bth|\lambda_{\min}, \sigma_{\mc A}), & \lambda = \lambda_{\min}.  
  \end{cases}
\end{align}

At this stage, we introduce the following lemma to address problem \eqref{prob:pararl} through a typical homotopy optimization method \cite{dunlavy2005homotopy}.

\begin{Lemma}
\label{Lemmaa:1}
  On the basis of homotopy optimization, one can initialize a sequence of positive values, i.e., $\delta^i, i = 1, \cdots, I$,  subject to:
\begin{align}
   \sum_{i = 1}^{I} \delta^i = -\lambda_{\min}, \label{eq:homcon} 
 \end{align} 
and also initialize a point $(\bth^0, \lambda^0)$, where  $\bth^0$ denotes a (local) optimizer of $J_{\h}(\bth| \lambda^0, \sigma_{\mc A})$ and  $\lambda^0 = \lambda_{\min}$; and then iterate the following update:
\begin{align}
   \lambda^{i} =  \lambda^{i-1} + \delta^{i}, \label{eq:lbdupdate}
\end{align}
and calculate a local optimizer $\bth^{i}$ of $ J_{\h}(\bth |\lambda^i, \sigma_{\mc A})$ by using gradient descent starting from $\bth^{i-1}$. Eventually, this homotopy approach is able to result in point $(\bth^{I}, 0)$, where $\bth^{I}$ is a local minimizer of problem \eqref{prob:pararl} \cite{dunlavy2005homotopy}.
\end{Lemma}

The motivation of the homotopy optimization approach is follows: Starting with a sufficient small value of $\lambda_{\min} < 0$, a very large cumulative penalty $|\lambda_{\min} G|$ may penalize the corresponding policy (parametrized by ${\bth^0}$) to generate intended actions, e.g., caching decisions that fully exploit available caching storage in the considered problem. Thereafter, by using homotopy optimization, we attempt to carefully tune policy parameter ${\bth^0}$ to a (local) optimizer $\bth^{I}$ of the original problem \eqref{prob:pararl}, which is likely to produce good decisions despite of applying mapping function $\siga(\cdot)$.
 
\subsection{Proposed DRL Algorithm}
In each iteration of homotopy optimization, computing an optimizer (e.g.,$\bth^{i}$) of problem \eqref{prob:hrl} offline is somehow impractical under an unknown temporally evolving environment. For this reason, we custom-build HDDPG for problem \eqref{prob:pararl} by recasting the basic elements of DRL and unfolding the iterative procedure of homotopy optimization introduced in Lemma \ref{Lemmaa:1}. 

Specifically, as an actor-critic approach, HDDPG maintains a parametrized critic $Q_{\h, \bphi} (\bm S, \bm A)$ and actor $\mu_{\bth} (\bm S)$ in addition to a mapping function $\siga(\cdot)$, where a feasible policy is given by $\pi_{\bth}(\bm S) = \siga[\mu_{\bth} (\bm S)]$. Following the sketch of a plain DDPG,\footnote{For clarity, we term the DRL algorithm proposed in \cite{Lillicrap2015ContinuousCW}  as plain DDPG.} we introduce the proposed algorithm as follows. Evidently, the objective $J_{\h} (\bth| \lambda, \sigma_{\mc A})$ can be equivalently reformulated as:
\begin{align}
  J_{\h} (\bth| \lambda, \sigma_{\mc A}) = \mathbb{E} \left[\sum_{t= 0}^{+\infty} (\gamma)^{t} \left( R^{t+1} + \lambda g(\bm S^{t}|\siga) \right) \right].
\end{align}
Accordingly, we define the homotopy reward after taking action $\bm A^t$ as:  
\begin{align}
  R_\h^{t+1} =  R_\h(\bm S^{t+1}, \bm S^t, \bm A^t ) \triangleq  R^{t+1} + \lambda g(\bm S^t|\siga), \label{eq:hreward}
\end{align}
which is known at epoch $t+1$. 
Then, the homotopy Q-function can be given by: 
\begin{align}
  Q_\h(\bm S^t, \bm A^t) = \mathbb{E} \left[ \sum_{\tau = 0}^{+\infty} (\gamma)^{ t + \tau} R_{\h}^{t + \tau+1} \bigg|\bm S^t, \bm A^t, \mu_{\bth}, \siga \right], 
\end{align}
which implies the discount cumulative homotopy reward after taking an action $\bm A^t$ under state $\bm S^t$ and thereafter following policy $\pi_{\bth}(\cdot) = \siga[\mu_{\bth} (\cdot)]$. 
As a direct deduction of the Bellman optimality equation \cite{sutton1998introduction}, we have the following {\it Homotopy Bellman Optimality Equation}. 
\begin{Lemma}
\label{lamma:2}
  An optimal $Q_\h^*(\bm S, \bm A)$ satisfies the following recursive equality:
  \begin{align}
    Q_\h^*(\bm S, \bm A) = \mathbb{E}_{R_\h, \bm S'|\bm S,\bm A} \left[ R_\h + \gamma \max_{\bm A' \in \mc A} Q^*_{\h} (\bm S', \bm A')  \right],
  \end{align}
  where $\bm S'$ denotes the subsequent state after taking an optimal action $\bm A$; and $R_\h$ denotes the associated homotopy reward. 
\end{Lemma}
Accordingly, to estimate an optimal homotopy Q-function, the critic $ Q_{\h, \phi}(\bm S, \bm A)$ can be learned by using Lemma \ref{lamma:2}. Specifically, we update $\bphi$ by minimizing the following loss function:
\begin{align}
  Loss\left(\bphi\right) =  \mathbb{E}_{\xi_\h}  \left[ \left(y - Q_{\h, \phi}(\bm S, \bm A)\right)^2 \right], \label{eq:loss0}
\end{align}
where $\xi_{\h} = (\bm S, \bm A, R_{\h}, \bm S')$; and $y$ denotes the target value: 
\begin{align}
  y = R_\h + \gamma Q_{\h, \phi}\left(\bm S', \mb A' \right)\big|_{\bm A' = \pi_{\bth}(\bm S')}. \label{eq:target1}
\end{align}

Regarding the update of the actor, it depends on the gradient of the objective $J_\h (\bth| \lambda, \siga )$, which brings us to the following the {\it Deterministic Policy Gradient Theorem for HDDPG}. 
\begin{Lemma}
\label{prop:homgrad} 
Consider a homotopy deep deterministic policy with a continuous action space $\mc A$ and a homotopy variable $\lambda$, as well as a mapping function $\sigma_{\mc A}$. Suppose that $\siga(\cdot)$ is continuous. Then, the deterministic policy gradient exists when $\nabla_{\bth} \mu(\bm S)$ and $\nabla_{\bm A} Q_{\h, \bphi} (\bm S, \bm A)$ exist, i.e.:  
\begin{align}
   \nabla_{\bth} J_{\h} (\bth| \lambda, \sigma_{\mc A}) \triangleq \mathbb{E}_{\bm S}     \left[\nabla_{\bm A} Q_{\h, \phi} (\bm S, \bm A)\big|_{\bm A = \siga[\mu_{\bth} (\bm S)]}    \nabla_{\bm A'} \siga( \bm A')|_{\bm A' = \mu_{\bth}(\bm S)}       \nabla_{\bth} \mu_{\bth} (\bm S) \right]. \label{eq:homgrad}
\end{align}
\begin{proof}
  See Appendix \ref{appen:B}.
\end{proof} 
\end{Lemma}
Finally, we leverage inexact gradient descent methods to update $\{\bth, \bphi, \lambda\}$ \cite{cassioli2013convergence}. In particular, the updates of $\bphi, \bth$ occur at each epoch, i.e.:
\begin{align}
  \bphi \leftarrow \bphi - \alpha_c \nabla_{\bphi} Loss(\bphi), \label{eq:up3} \\
     \bth \leftarrow \bth + \alpha_a \nabla_{\bth} J_{\h} (\bth| \lambda, \sigma_{\mc A}) , \label{eq:up4}
\end{align} 
where $\alpha_{c}$ and $\alpha_a$ are the learning rates of the critic and actor, respectively; and  $\lambda$ can be updated by a slow circle, i.e., 
after every $I_0$ epochs, one can execute the following:
\begin{align}
  \lambda^{iI_0} \leftarrow \lambda^{(i-1)I_0} + \delta^{i}, \label{eq:lbd1}
\end{align}
where sequence $\{\lambda^i\}_{i=1}^{I}$ should meet the equality in \eqref{eq:homcon}. 

\begin{Remark}
   In contrast with plain DDPG, which constitutes a special case of the proposed HDDPG, i.e., $\lambda = 0$, properly introducing a penalty term into the objective function assists to infer which actions should be better to take and avoid becoming stuck in suboptimal solutions.
   More importantly, we unfold the homotopy optimization approach in Lemma \ref{Lemmaa:1} into a DRL, which can be done through interacting with environments. 
\end{Remark}
In the ensuing sections, we will apply HDDPG to the cooperative coded caching problem and propose a centralized caching design, and further generalize HDDPG in decentralized settings to reduce complexity and communication cost.

\section{Centralized HDDPG-based Cooperative Coded Caching}
In this section, we introduce a centralized HDDPG (C-HDDPG) design for multi-agent cooperative coded caching. As illustrated in Fig. \ref{fig:stru}(a), the system operation is at the level of centralized control. The CP serves as a centralized agent and coordinates the cooperative caching policies for all SBSs based on global information. To realize this, the CP maintains a (centralized) critic network $Q_{\bphi} (\bm S, \bm A)$ and a (centralized) actor network $\mu_{\bth}(\bm S)$, as well as a mapping function $\siga(\cdot)$. In what follows, we first introduce a detailed implementation of the proposed centralized design, and then analyze its communication overhead and complexity. 

\begin{figure*}
\centering
  \begin{subfigure}[t]{.31\linewidth}
  \centering
  \includegraphics[width=5cm, height=4cm]{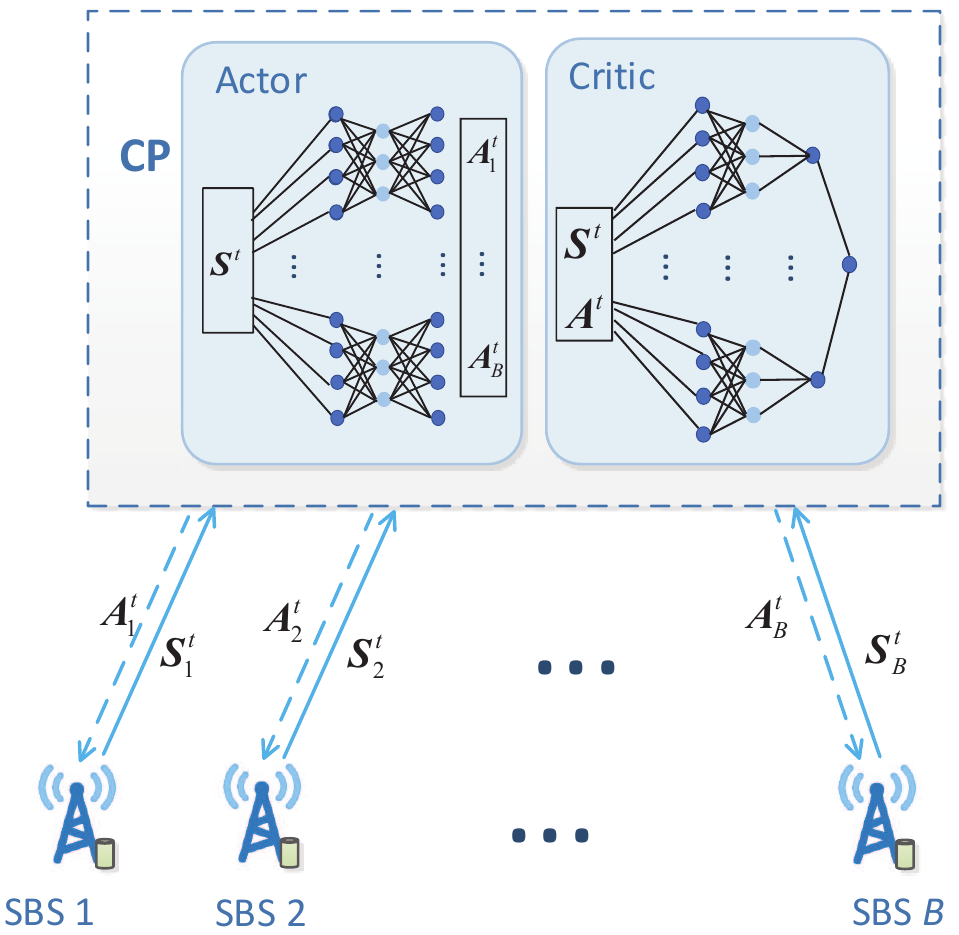}
  \caption{Centralized control.}
  \label{fig:stru-cen}
\end{subfigure}
\begin{subfigure}[t]{.31\linewidth}
  \centering
  \includegraphics[width=5cm, height=4cm]{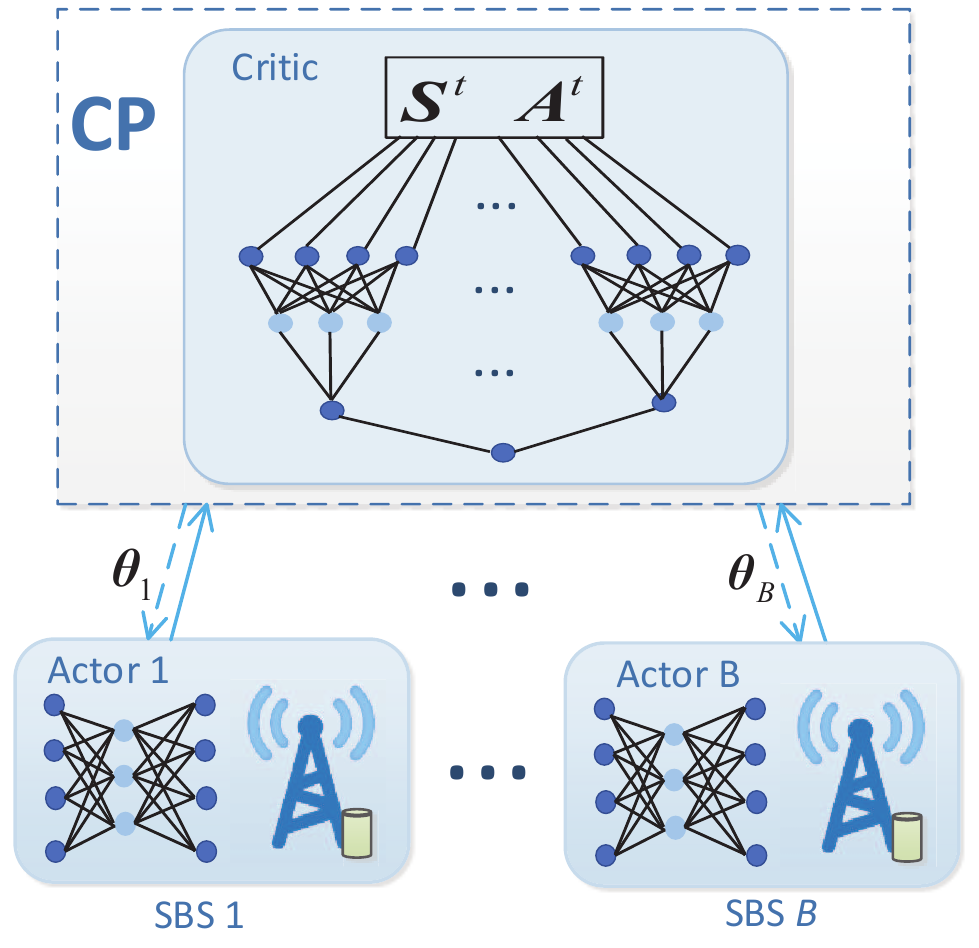}
  \caption{Partially decentralized control.}
  \label{fig:stru-ma}
\end{subfigure}
\begin{subfigure}[t]{.31\linewidth}
  \centering
  \includegraphics[width=5cm, height=4cm]{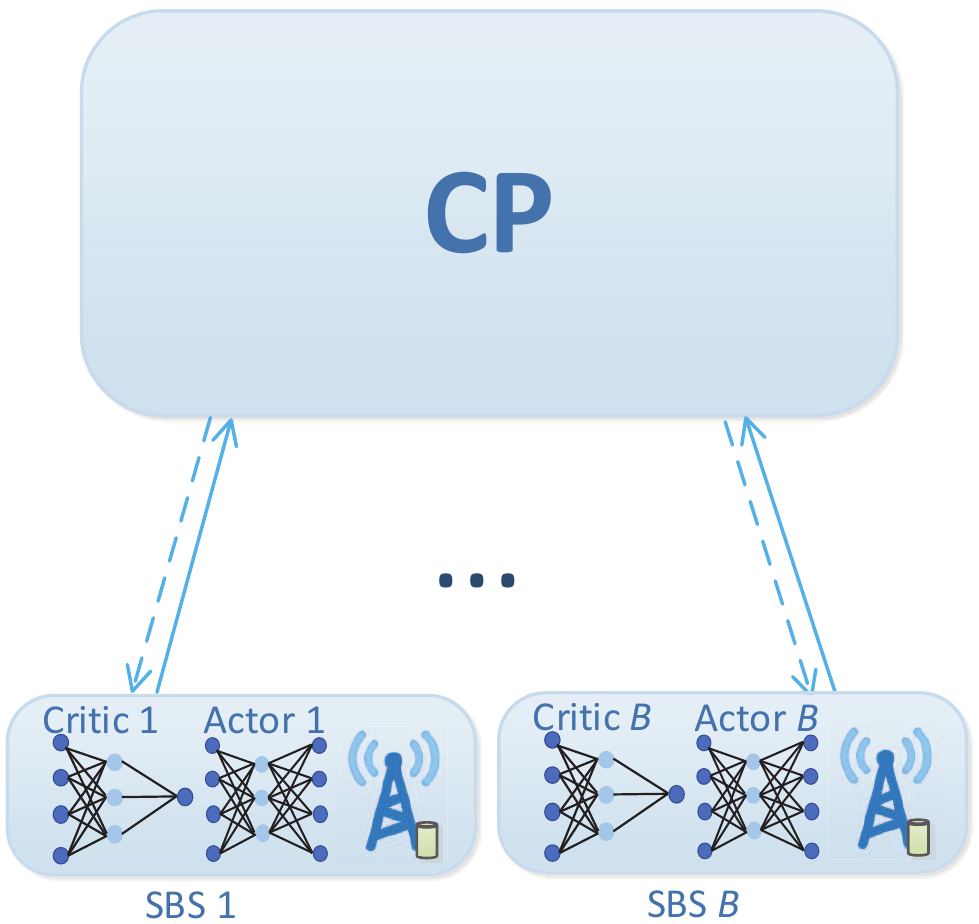}
  \caption{Fully decentralized control.}
  \label{fig:stru-fd}
\end{subfigure}
\caption{Proposed multi-agent DRL-based cooperative coded caching framework.}
\label{fig:stru}
\end{figure*}
\subsection{Proposed Centralized HDDPG-based Design}

The system operation includes two procedures, i.e., network training and network evaluation. In general,  during network evaluation, the CP simply leverages the actor and mapping function to make caching decisions, while the critic is only necessary during training procedure to fine-tune the actor. 
The details of network design and training procedure are introduced as follows.

\subsubsection*{Network Design} In general, both networks, i.e., critic and actor, can be implemented by fully connected DNNs where each hidden layer has a batch of neurons and an activation function to perform nonlinear transformations \cite{lecun2015deep}.
The output of the critic should be a scalar, which corresponds to the estimated value of the Q-function.  
To generate feasible actions, we elaborate on how to design the actor network $\mu_{\bth}$ and mapping function $\siga$.
It is evident that the number of neurons in the output layer of  $\mu_{\bth}$ should match the dimension of a joint action, i.e., $F\times B$ (and these neurons output a long vector $\bm z = [z_{1,1}, z_{2,1}, \cdots, z_{f,b}, \cdots, z_{F,B}]$). Then, we use the following activation function (e.g., realized by {\it Scaling} and {\it SoftMax}) to refine $\bm z$, i.e.:
\begin{align}
  \sigma_{f,b} (\bm z) = L \times  \frac{\exp(z_{f,b})}{\sum_{ f' \in \mc F} \exp(z_{f',b})}, \forall f \in \mc F, \forall b \in \mc B,
\end{align}
which thereafter is filtered by a mapping function\footnote{For instance, $\min\{\mb 1, \bm X\}$ is an element-wise operator that executes $\max \{1, x_{f,b}\}$ for any element $x_{f,b}$ of $\bm X$.} $\siga (\cdot)= \min\{\bm 1_{F\times B}, \cdot\}$. Accordingly, any proto-action $\mu_{\bth} (\bm S)$ can be mapped into a feasible action, i.e., $ \min\{\bm 1_{F\times B}, \mu_{\bth} (\bm S) \}$. 

\subsubsection*{Update} To proceed, the technique of {\it Replay Buffer (RB)} $\Xi$ is introduced to store historical experiences $\xi^t = (\bm S^t, \bm A^t, R_\h^{t+1}, \bm S^{t+1})$, which serves as the data set for network training. The buffer size $|\Xi|$ is usually finite, and thus the most outdated experience should be replaced by the current one as long as $\Xi$ is fully loaded.
Subsequently, at each epoch, we can randomly sample a mini-batch of $N$ experiences (e.g., set $\Xi_N $) from {\it RB} to update parameters of the critic and the actor networks. More concretely, 
parameter $\bphi$ of the critic network can be updated by minimizing the following loss function:
\begin{align}
  Loss\left(\bphi\right) =  \mathbb{E}_{\xi^t \sim \Xi_N}  \left[\left( y_{\bphi^-}^t  - Q_{\h, \phi}(\bm S^{t}, \bm A^{t})\right)^2 \right], \label{eq:loss}
\end{align}
where the expectation is over all of the sampled experiences; $y_{\bphi^-}^t$ denotes the target value:
\begin{align}
  R_\h^{t+1} + \gamma Q_{\h, \phi^-}\left(\bm S^{t+1}, \siga[\mu_{\bth^-} (\bm S^{t+1})]\right), \label{eq:target}
\end{align}
and $Q_{\h, \bphi^-}(\bm S, \bm A)$ and $\mu_{\bth^-}(\bm S)$ denote the target critic and the target actor with parameters $\bphi^-$ and $\bth^-$, respectively. To stabilize training \cite{Lillicrap2015ContinuousCW}, target networks should be slowly updated, i.e.:
\begin{align}
  \bphi^- \leftarrow \tau \bphi + (1 - \tau) \bphi^-, \label{eq:up1}\\
  \bth^- \leftarrow \tau \bth + (1 - \tau) \bth^-, \label{eq:up2}
\end{align}
where $\tau$ is a very small step size. 
With regard to updating parameter $\bth$, 
the corresponding homotopy deterministic policy gradient $\nabla_{\bth} J_{\h}$ can be estimated by \eqref{eq:homgrad}. 
In addition, the homotopy variable should be updated according to \eqref{eq:lbd1}. 

\subsubsection*{Exploration}
To avoid becoming stuck in suboptimal policies, exploration is usually needed during network training. 
The purpose of this process is to gather sufficient experiences, which then are used to infer what actions should be adopted under different states.
In continuous decision-making applications, a typical method is to add Ornstein-Uhlenbeck (OU) random noise to the action generated by the actor \cite{Lillicrap2015ContinuousCW}, i.e.:
\begin{align}
  \pi_{\bth}^{\rm{explore}} (\bm S^t) = \sigma'_{\mc A} \left( \pi_{\bth} (\bm S^t) + \beta^t \Delta^t \right), 
\end{align}
where $\siga'$ is a simple mapping function if the noise-perturbed action violates ${\mc A}$;
 $\Delta^t = [\delta^t_{f,b}] \in \mathbb{R}^{F\times B}$; and each element $\delta_{f,b}^t$ denotes a sample drawn from a continuous OU process \cite{Lillicrap2015ContinuousCW}; and $\beta^t \geq 0$ is a diminishing parameter. 

To this end, an entire implementation of this centralized control is shown in Algorithm \ref{alg:C-HDDPG}. 
\subsection{Fronthaul Communication Complexity} In the proposed centralized caching design, the CP needs to frequently communicate with SBSs during network training and evaluation. Herein, we briefly analyze fronthaul communication complexity of this centralized control, which is described by the total dimension of variables that are transmitted between the CP and SBSs. We consider the worst case, in which each SBS is fully loaded and serves a maximum number of users, e.g.,  $|\mc K_b^t| = K$. 
Specifically, during network training, the CP needs to obtain information about the system state $[\{f_k^t\}, \mc E^t, \bm L^t]$ at each epoch. The dimension of user requests should be $BK$.
Network connectivity $\mc E^t$ can be computed by knowing the coordinates of the active users; by denoting the coordinates as two dimensional vectors, the total dimension of user positions is $2BK$. 
Clearly, the CP has the exact information about cache allocation $\bm L^t$, which is termed as $\bm A^{t-1}$ in its {\it RB}; thus, no fronthaul cost is involved. 
Afterwards, the CP uses the the fronthaul to inform caching decisions $\{\bm A_b^t\}$; the total dimension of the involved variables is given by $BF$. Regarding reward $R^{t+1}$, it can be inferred from state $\bm S^{t+1}$.
Hence, the overall fronthaul communication complexity during network training is $\mc{O} (3BK + BF)$.
When the system runs in an evaluation procedure, the CP again needs to know the system state and inform each SBS of its caching decision. Consequently, the corresponding fronthaul communication complexity is $\mc{O} (3BK + BF)$.

Moreover, the critic and actor are built upon system states and joint actions, i.e., $(\bm S^t, \bm A^t)$, which is in the order of $\mc{O} (B^2)$ of local observations and actions, i.e., $(\bm S^t_b, \bm A_b^t)$. For this reason, the computational complexity would be excessively high as the number of agents increases for a continuous RL problem \cite{lowe2017multi}. 
To address this issue, we now focus on developing efficient decentralized algorithms in following sections.

\begin{algorithm}[!t]
  \caption{Proposed C-HDDPG-based Cooperative Coded Caching}\label{alg:C-HDDPG}
  \begin{algorithmic}[1]
    \State Initialize $ \tau $, $ \alpha_c $, $ \alpha_a $, $ \gamma $,  $ N $, $ I_0 $, $\lambda = \lambda_{\min}$
    \State Initialize parameter $\bphi$ for critic network and $ \bth $ for actor network
    \State Initialize parameters $\bphi^- \leftarrow \bphi$, $\bth^- \leftarrow \bth$ for target critic network and target actor network
    \State Initialize {\it RB} $\Xi$ and mapping function $\siga$
    \State Initialize $\delta^1, \delta^2, \cdots, \delta^I$ 
    \For{$ t=0,1,2, \cdots $}
    \State Input $\bm S^t$ to actor and output $\bm A^t = \pi_{\bth}^{\rm explore} (\bm S^t)$
    \State Take action $\bm A^t $ and observe $\bm S^{t+1} $, $ R^{t+1} $
    \State Calculate $ R^{t+1}_\h $ by \eqref{eq:hreward}
    \State Store $\xi^t =  \left( \bm S^t, \bm A^t, R^{t+1}_\h, \bm S^{t+1} \right) $ into {\it RB}
    \Procedure{TrainHDDPG}{}
	    \State Randomly sample a mini-batch of $N$ experiences from relay buffer as $ \Xi_N $
      \State Update $\bphi, \bth$ by \eqref{eq:up3} and \eqref{eq:up4}, respectively
	    \State Update $\bphi^-$, $\bth^-$ by \eqref{eq:up1} and \eqref{eq:up2}, respectively
      \If{$t == i\times I_0$} $\lambda \leftarrow \lambda + \delta^i$ 
      \EndIf
	\EndProcedure    
    \EndFor 
  \end{algorithmic}
\end{algorithm}

\section{Partially Decentralized HDDPG-based Cooperative Coded Caching}
In this section, to circumvent excessive communication cost and high complexity in the centralized design, we develop a partially decentralized (PD)-HDDPG-based cooperative coded caching design. This scheme operates at the level of PD control, in the sense that a (centralized) critic is used to train (local) actors that separately approximate the caching policy of each SBS.
\begin{algorithm}[!t]
  \caption{Proposed PD-HDDPG-based Cooperative Coded Caching}\label{alg:PD-HDDPGA}
  \begin{algorithmic}[1]
    \State Initialize $ \tau $, $ \alpha_c $, $ \alpha_a $, $ \gamma $,  $ N $, $ I_0$, $ \lambda = \lambda_{\min}$
    \State Initialize parameter $\bphi$ for the critic  and $\bth = \{\bth_b, \forall b \in \mc B\}$ for actors
    \State Initialize parameters $\bphi^- \leftarrow \bphi$, $\bth^- \leftarrow \bth$ for target critic and target actors
    \State Initialize {\it RB} $ \Xi $ and mapping functions $\{\sigma_b, \forall b \in \mc B\}$
    \State Initialize $\delta^1, \delta^2, \cdots, \delta^I$
    \For{$ t=0,1,2, \cdots $}
      \For{$b \in \mc B$}
        \State Observe $\bm S_b^t$ and compute action $\bm A_b^t $ through $\pi_b$ with proper exploration 
      \EndFor
    \State Execute $\{\bm A_b^t, \forall b\}$ in a real environment and observe $\bm S^{t+1} = \{\bm S_1^{t+1}, \cdots, \bm S_B^{t+1}\}  $, $ R^{t+1} $
    \State Let $\mb A^t = \{\bm A_1^t, \cdots, \bm A_b^t\}$ and calculate $R_{\h}^{t+1}$ by \eqref{eq:hreward}
    \State The CP pushes $\xi^t =  \left( \bm S^t, \bm A^t, R_{\h}^{t+1}, \bm S^{t+1} \right) $ into {\it RB}
    \Procedure{TrainHDDPG}{}
    \State Randomly sample a mini-batch of experiences $ \Xi_N $
    \State Calculate $\nabla_{\bphi} Loss(\bphi)$ and $\bphi \leftarrow \bphi - \alpha_c \nabla_{\bphi} Loss(\bphi) $
    \For{ $ b \in \mc B$}
      \State $\bth_b \leftarrow \bth_b + \alpha_a \nabla_{\bth_b} J_{\h}  $
    \EndFor
    \State $\bphi^-  \leftarrow (1-\tau)\bphi^- + \tau \bphi$ 
    , ~$\bth^- \leftarrow (1-\tau)\bth^- + \tau \bth$
    \If{$t == i\times I_0$} $\lambda \leftarrow \lambda + \delta^i$ \EndIf
    \EndProcedure
    \EndFor 
  \end{algorithmic}
\end{algorithm}
\subsection{Partially Decentralized Multi-Agent HDDPG}
In a PD multi-agent framework, each agent maintains an actor and mapping function to produce its actions. To augment collaboration among multiple agents, these actors are trained with the aid of a (centralized) critic. 
Specifically, agent $b$ has an actor $\mu_b$ (parametrized by $\bth_b$) and mapping function\footnote{For simplicity of notation, $\mu_b$ and $\sigma_b$ are abbreviations of $\mu_{\bth_b}$ and $\sigma_{\mc A_b}$, respectively.} $\sigma_b$, which are able to map a (local) proto-action $\mu_b(\bm S_b)$ into the corresponding action space $\mc A_b$. 
Accordingly, the policy function for agent $b$ can be expressed by $\pi_b = \sigma_b[\mu_b(\cdot)]$.

On the basis of homotopy optimization, all agents cooperatively seek polices to jointly maximize the following homotopy function:
\begin{align}
  J_\h(\bth_1, \cdots, \bth_B|\lambda, \siga) = \mathbb{E} \left[\sum_{t = 0}^{\infty} (\gamma)^{t} R_\h^{t + 1} \big| \mu_{\bth}, \siga \right],
\end{align}
where we define $\mu_{\bth} \triangleq \{\mu_1, \cdots, \mu_B\}$ and $\siga \triangleq \{\sigma_1, \cdots, \sigma_B\}$; and the homotopy reward $R_\h^{t}$ can be given by \eqref{eq:hreward}. Next, a (centralized) critic $Q_{\h, \bphi}(\bm S^t, \bm A_1^t, \bm A_2^t, \cdots, \bm A_b^t) $ is leveraged to estimate 
$
  \mathbb{E} \left[ \sum_{\tau = 0}^{\infty} (\gamma)^{t+\tau} R_\h^{t + \tau +1} \big| \bm S^t, \bm A_b^t, \mu_b, \sigma_b, \forall b \in \mc B\right].
$
Similar to C-HDDPG, parameter $\bphi$ can be learned by minimizing the following loss function:
\begin{align}
  Loss\left(\bphi\right) = \mathbb{E}_{ \xi_t \sim \Xi_N} \left[ \left(y_{\bphi^-}^t  - Q_{\h, \phi}(\bm S^{t}, \bm A_1^t, \cdots, \bm A_b^t) \right)^2\right], \label{eq:loss2}
\end{align}
where $y_{\bphi^-}^t$ denotes the target value, i.e., 
$
  R_\h^{t+1} + \lambda Q_{\h, \bphi^-}\left(\bm S^{t+1}, \bm A_1, \cdots, \bm A_b\right) \big|_{\bm A_b = \sigma_b[\mu_b^-(\bm S_b^{t+1})]}, \label{eq:target2}
$
and $\bth_b^-, \bphi^-$ are parameters of target actor $\mu_b^-$ and target critic $Q_{\h, \bphi^-}$, respectively. Furthermore, the gradient for training parameter $\bth_b$ can be approximated by:
 \begin{align}
  \nabla_{\bth_b} J_\h \approx \mathbb{E}_{ \xi_t \sim \Xi_N} \left[\nabla_{\bm A_b} Q_{\h,\bphi}(\bm S^{t}, \bm A_1, \cdots, \bm A_b)|_{\bm A_b = \pi_b (\bm S_b^t)} \nabla_{\bth_b} \pi_{b}(\bm S_b^t)\right]. \label{eq:jgrad2}
\end{align}
Similarly, homotopy variable $\lambda$ should be updated in accordance with \eqref{eq:lbd1}.

\subsection{Implementation}
As depicted in Fig. \ref{fig:stru} (b), we propose a PD-HDDPG-based cooperative coded caching design. Specifically, the CP maintains a centralized (critic), while each SBS has a local actor and mapping function. In addition, the actor and mapping function are designed in the same manner as the centralized scheme to ensure that their outputs are feasible to \eqref{eq:action_space}. Particularly, $\sigma_b (\cdot) = \min \{\bm 1_{F}, \cdot\}, \forall b \in \mc B$.
During the training procedure, the critic and (local) actors should be learned in the CP. We again adopt the techniques of exploration and {\it RB}, and the entire procedure is similar to what we have presented in Algorithm \ref{alg:C-HDDPG}. The detailed implementation is shown in Algorithm \ref{alg:PD-HDDPGA}. Notably, after all actors are fine-tuned, the CP needs to send actor parameters (e.g., $\bth_b, \forall b$) to SBSs, which thereafter can locally compute actions.


\subsection{Fronthaul Communication Complexity} 
During the training procedure, fronthaul communication complexity is the same as that of C-HDDPG, i.e., $\mc{O} (3BK + BF)$. When the system runs in an evaluation procedure, each SBS computes its action locally; obviously,  no fronthaul communication is incurred when observing content requests and positions of local users. 

\section{Fully Decentralized HDDPG-based Cooperative Coded Caching}
To further reduce complexity and fronthaul signaling, we propose a fully decentralized (FD) control for cooperative coded caching. Particularly, each SBS serves as an independent learner to locally train its caching policy. Hereunder, we first present the FD-HDDPG-based caching design, and then briefly summarize the complexity of all of the proposed designs.

\begin{algorithm}[!t]
  \caption{Proposed FD-HDDPG-based Cooperative Coded Caching}\label{alg:FD-HDDPGA}
  \begin{algorithmic}[1]
    \State Initialize $ \tau $, $ \alpha_c $, $ \alpha_a $, $ \gamma $, $I_0$, $\lambda = \lambda_{\min}$
    \State Initialize parameters $\bphi = \{ \bphi_b, \forall b \in \mc B\} $ for critics and $\bth = \{\bth_b, \forall b \in \mc B\}$ for actors
    \State Initialize parameters $\bphi^- \leftarrow \bphi$, $\bth^- \leftarrow \bth$ for target critics and target actors
    \State Initialize $ \Xi_b, N_b, \forall b \in \mc B$ and $\{\sigma_b, \forall b \in \mc B\}$
    \State Initialize $ \delta^1, \cdots, \delta^I$
    \For{$ t=0,1,2, \cdots $}
      \For{$b \in \mc B$}
        \State Execute $\bm A_b^t = \pi_b(\bm  S_b^t)$ with proper exploration and calculate $g_b(\bm  S_b^t) \leftarrow L - \| \bm A_b^t\|_1$
      \EndFor
      \State Observe $\bm  S_b^{t+1}, \forall b \in \mc B$, and $R^{t+1}$ 
      \State Calculate $ R_\h^{t+1}$ by \eqref{eq:hreward3}
      \For{$b \in \mc B$}
        \State SBS $b$ stores $\xi^t =  \left( \bm  S_b^t, \bm A_b^t, R_\h^{t+1}, \bm  S_b^{t+1} \right) $ into $\Xi_b$
      \EndFor
    \Procedure{TrainHDDPG}{}
    \For{$b \in \mc B$}
      \State Randomly sample a mini-batch of experiences  $ \Xi_{b,N} $
      \State $\bphi_b \leftarrow \bphi_b - \alpha_c \nabla_{\bphi_b} Loss(\bphi_b) $
      \State $\bth_b \leftarrow \bth_b + \alpha_a \nabla_{\bth_b} J_\h$
      \State $\bphi_b^-  \leftarrow (1-\tau)\bphi_b^- + \tau \bphi_b$,  $\bth_b^- \leftarrow (1-\tau)\bth_b^- + \tau \bth_b$
          \If{$t == i\times I_0$} $\lambda \leftarrow \lambda + \delta^i$ \EndIf
      \EndFor
    \EndProcedure
    \EndFor 
  \end{algorithmic}
\end{algorithm}
\subsection{Fully Decentralized Cooperative Coded Caching Design}
As shown in Fig \ref{fig:stru}(c), each SBS has a set of critic, actor, and mapping functions. These basic elements are designed in the same manner as that of C-HDDPG, but built upon local observations.  
More precisely, with an actor $\mu_b(\cdot)$ and mapping function $\sigma_b$, SBS $b$ can obtain a feasible action by mapping a proto-action $\mu_b(\bm S_b)$ into the action space $\mc A_b$, i.e., $\sigma_b[\mu_b(\bm S_b)]$. SBSs are encouraged to cooperate with each other and receive a common reward $R^t$ from the environment as the performance criterion to evaluate their policies. On the basis of homotopy optimization, a (local) critic $Q_{\h, \bphi_b}(\bm  S_b^t, \bm A_b^t)$ is designed for SBS $b$ to estimate $\mathbb{E} \left[ \sum_{\tau = 0}^{+\infty} (\gamma)^{t+\tau} R_\h^{t+\tau+1}|\bm S_b^t, \bm A_b^t, \mu_b, \sigma_b\right],$
where $R_\h^t$ is defined as:
\begin{align}
  R_\h^t \triangleq R^{t} + \lambda \sum_{b \in \mc B} g_{b} (\mb  S_b^{t-1}|\sigma_b), \label{eq:hreward3}
\end{align}
and $g_{b} (\mb  S_b^{t-1}|\sigma_b) \triangleq L - \| \sigma_b(\mu_b(\bm  S_b^{t-1})) \|_1, \forall b$. Subsequently, each agent is envisioned to independently train its critic and actor. In addition,
the training procedure should follow the same workflow  as C-HDDPG, which is presented in Algorithm \ref{alg:FD-HDDPGA} in greater detail.
\subsection{Fronthaul Communication Complexity} 
During network training, although fronthaul communications are not necessary for SBSs to obtain local observations, each SBS still needs to know the homotopy reward. SBSs first locally computes $\{g_b(\bm  S_b^t|\sigma_b)\}$, which are then aggregated by the CP and subsequently sent back to each SBS. Therefore, fronthaul communication complexity  during training is given by $\mc{O}\left(2B\right)$. For network evaluation, SBSs can directly calculate local actions according to their observations; thus, no fronthaul communication is needed. 
\begin{table}[h!]
\centering
\caption{Fronthaul communication complexity}
\begin{tabular}{c c c c c}
  \hline
        & C-HDDPG & PD-HDDPG & FD-HDDPG \\
    \hline
      Training &  $\mc{O} (3BK + BF)$ &  $\mc{O} (3BK + BF)$  & $\mc{O}\left(2B\right)$ \\
    \hline
      Evaluation & $\mc{O} (3BK + BF)$  & 0 &  0 \\
  \hline
\end{tabular}
\label{table:fso}
\end{table}
\begin{Remark}
  As a comparison, we summarize fronthaul communication complexity of all algorithms in Table \ref{table:fso}.  It can be observed that C-HDDPG incurs the highest fronthaul communication complexity in either training or evaluation procedure since it manipulates system operation in the centralized control. PD-HDDPG requires the same order of signaling as that of C-HDDPG during training, which thereafter operates in a decentralized manner during network evaluation; thus, fronthaul communication complexity during the evaluation procedure is as low as that of FD-HDDPG. Indeed, FD-DDPG has the lowest fronthaul communication complexity during two procedures, which might compromise performance. Therefore, by developing different controls of caching design, the proposed framework is envisioned to possess advantages of superior performance, as well as scalability to large-scale systems.
\end{Remark}

\section{Performance Evaluations}
In this section, we present performance evaluations of the proposed DRL algorithms for cooperative coded caching under different scenarios. Specifically, we first provide simulation setup and then compare the proposed algorithms with baselines. Subsequently, we investigate the impacts of system parameters on the proposed algorithms.
\subsection{Simulation Setup}
Unless stated otherwise, we consider the following default settings: a SCN covers a square area of $[0,1]$ km $\times$ $[0,1]$ km; four SBSs are uniformly deployed in the region, each of which has a communication radius of $r_0 = 500$ m and can provide service for a maximum of 100 mobile users during each epoch; and mobile users are randomly distributed by following PPP with density $9.5\times 10^{-5}$ during each epoch. 
Moreover, user preferences towards content are considered to have multiple patterns, i.e., each preference pattern follows a Zipf distribution, i.e., $p_{f} = \zeta_{f}^{-\kappa}/\sum_{f' \in \mc F} \zeta_{f'}^{-\kappa}$, where $\zeta_{f}$ denotes the popularity rank of content item $f$, which is temporally and randomly evolving as time passes; $\kappa$ denotes a skewness factor and randomly takes a value from $\{0.5, 1, 1.5, 2\}$; 
a catalog of 20 content items are encoded by MDS codes; and each SBS has a {\it fractional caching capacity} $L/ F = 0.2$, which indicates that each SBS can fetch $20\%$ of the total content. 

\subsection{Convergence Behavior}
To analyze the proposed caching framework, we consider the following baselines:
\begin{itemize}
  \item {\bf Centralized Optimization-Cache Updating (CO-CU):} This is a centralized optimization-based design, which is performed in the CP \cite{liao2017coding,wu2020jointLongTermCacheUpdating}. Specifically, one can first estimate the probability of each content item that could be requested by users under the coverage of SBS $b$, i.e., $p_{f,b} = N_{f,b}/\sum_{f' \in \mc F}N_{f',b}, \forall f$, where $N_{f,b}$ denotes how many requests of content item $f$ that SBS $b$ receives at the current epoch. Then, similar to the cooperative caching problem formulated in \cite{liao2017coding}, the joint caching decision (e.g. $\bm A^t$) can be optimized by minimizing the expected fronthaul traffic loads together with cache updating cost. 

  \item {\bf Local Optimization-Cache Updating (LO-CU):} This scheme works at a level of decentralized control. Instead of optimizing joint action $\bm A^t$ via CP in CO-CU, each SBS separately calculates its caching decision (e.g, $\bm A_b^t$) by minimizing the expected fronthaul loads corresponding to local user requests from its communication range. 
  \item {\bf Random Cache Updating (RCU):} At each epoch, every SBS randomly updates its caching resource until it reaches storage limits.  
  \item {\bf Plain DDPG-Based Schemes:} To assess the effectiveness of the proposed DRL, we consider to implement plain DDPG in different levels of controls, i.e., {\bf C-DDPG, PD-DDPG and FD-DDPG}, each of which follows a similar idea to the proposed caching approach. 
\end{itemize}

To implement the proposed algorithms, each critic is designed as follows: there are three hidden layers, each of which contains 512 neurons. Each actor consists of three hidden layers with 256, 128, and 64 neurons, respectively. All of the networks are trained by the {\it Adam} optimizer with a {\it polynomial learning rate} (e.g., readers are referred to \cite{mishra2019polynomial} for additional details), where we set initial learning rates for actors and critics as 0.01 and 0.001, respectively, and the power factor for decay as 0.9.
A mini-batch of 100 experiences are randomly sampled every time from {\it RB} that is capable of storing 5000 past experiences. 
Every target critic or target actor is updated by a step size $\tau = 0.001$, and the discount factor $\gamma = 0.99$. To perform policy explorations, we use an OU process with mean 0 and variance 1; the associated diminishing parameter $\beta$ is initialized as 0.9, and decreased at a rate of 0.995 every epoch until it reaches 0.0001. Finally, we initialize the following sequence to update the homotopy variable $\lambda$, i.e., $\{\delta^i = -0.1\times \lambda_{\min}, i = 1, 2, \cdots, 10\}$, where $\lambda_{\min} = -0.005$ and we update $\lambda$ every $I_0 = 1000$ epochs. 

As shown in Fig. \ref{Fig:Convergence-Cen-DDPG}, we first illustrate the learning curves of the proposed algorithm under centralized control. Particularly, we vary parameter $\lambda_{\min}$ to investigate the impacts of penalty. Each result is averaged over $N = 5000$ epochs, i.e., $\sum_{\tau=t-N+1}^{t} R^{\tau}/N$. 
It can be observed that in the first $10^4$ epochs, the curves of HDDPG-based algorithms rise markedly, and notable gaps can be observed between plain DDPG and HDDPG-based algorithms. In subsequent epochs, the learning curves of these DRL algorithms increase gradually until convergence. Clearly, when $\lambda_{\min} $ is -0.005 or -0.015, HDDPG-based algorithms achieve higher caching rewards than those of plain DDPG; when $\lambda_{\min}$ goes down to -0.1, the curve increases fairly slowly and converges to a level that is very close to plain DDPG. Therefore, if $\lambda_{\min}G$ is significantly large compared with the objective, it could dominant the actual objective, eventually leading to a suboptimal policy. These observations demonstrate that properly introducing the penalty term to RL (e.g., $\lambda_{\min}G$ in \eqref{prob:hrl}) could assist agents to infer better actions and speed-up convergence behavior. Furthermore, we propose to implement HDDPG (with $\lambda_{\min} = -0.005$) by initializing {\it RB} with 10\% warm-up experiences via optimization baselines (e.g., CU-CO) rather than the conventional exploration method used in plain DDPG, i.e., only utilizing OU random noise to explore action spaces. 
As can be seen, with few warm-up experiences, the proposed implementation can further improve performance compared with the conventional exploration under the same $\lambda_{\min} = -0.005$. This result implies that taking advantage of a good baseline improves the efficiency of exploration in DRL, yet at the cost of additional computational complexity. 

Fig. \ref{Fig:Convergence-MA-DDPG} and \ref{Fig:Convergence-FD-DDPG} show the learning curves of the proposed algorithms under partially and fully decentralized scenarios. As anticipated, the proposed PD-HDDPG and FD-HDDPG respectively outperform plain DDPG in both scenarios. It is worthing noting that, the shaded region around each learning curve shows reward deviations, which measure the robustness of each policy. Obviously, the DRL-based designs exhibit more centered results while the rewards, achieved by the optimization baseline, spread out over quite a broad range. This observation demonstrates the effectiveness of HDDPG to track and adapt to dynamic features of wireless networks.  

\begin{figure}
\begin{minipage}[t]{.495\linewidth}
  \centering
  \includegraphics[scale=0.38]{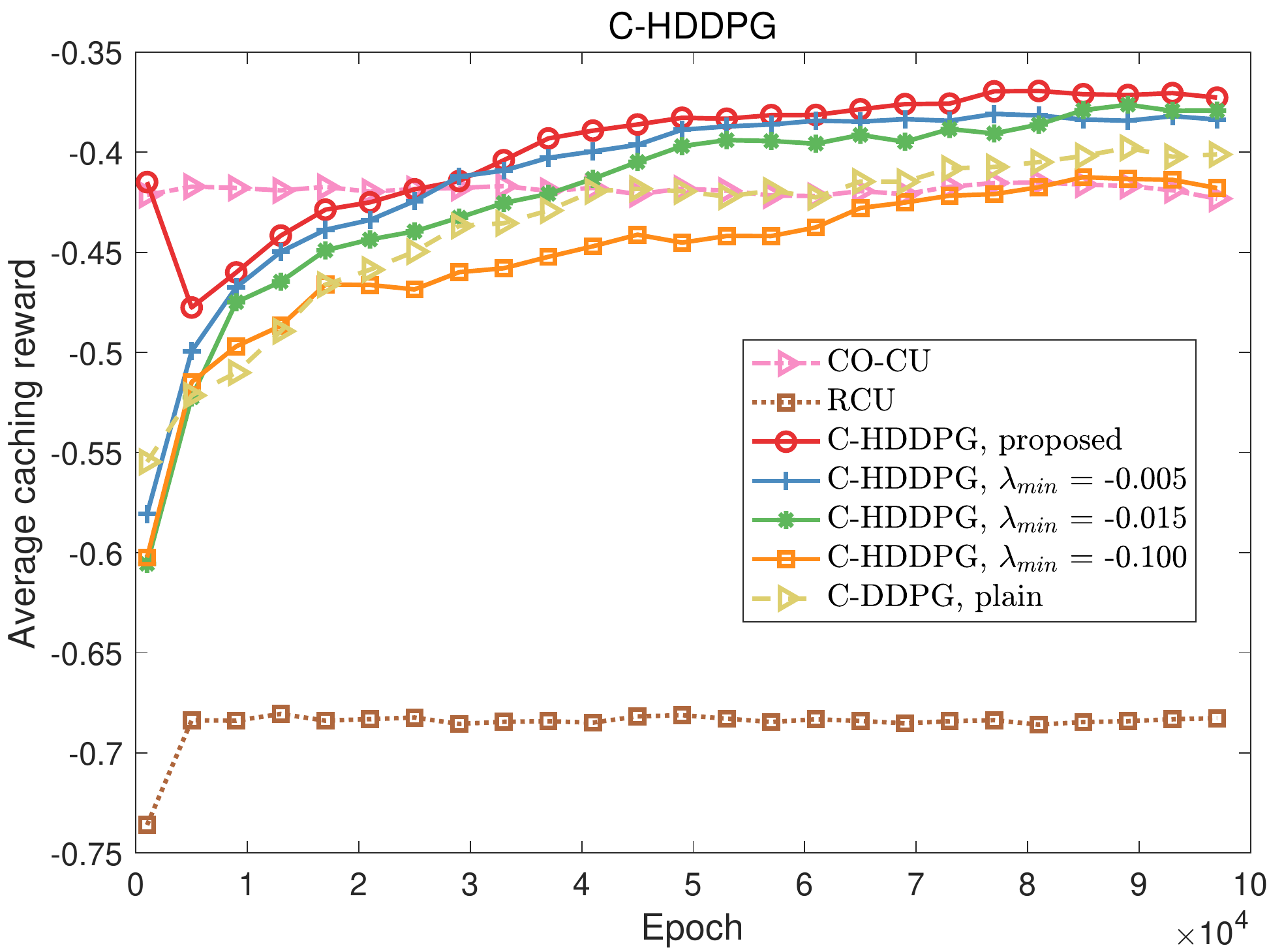}
  \caption{Learning curves of C-HDDPG. }
  \label{Fig:Convergence-Cen-DDPG}
\end{minipage}
\begin{minipage}[t]{.495\linewidth}
  \centering
  \includegraphics[scale=0.38]{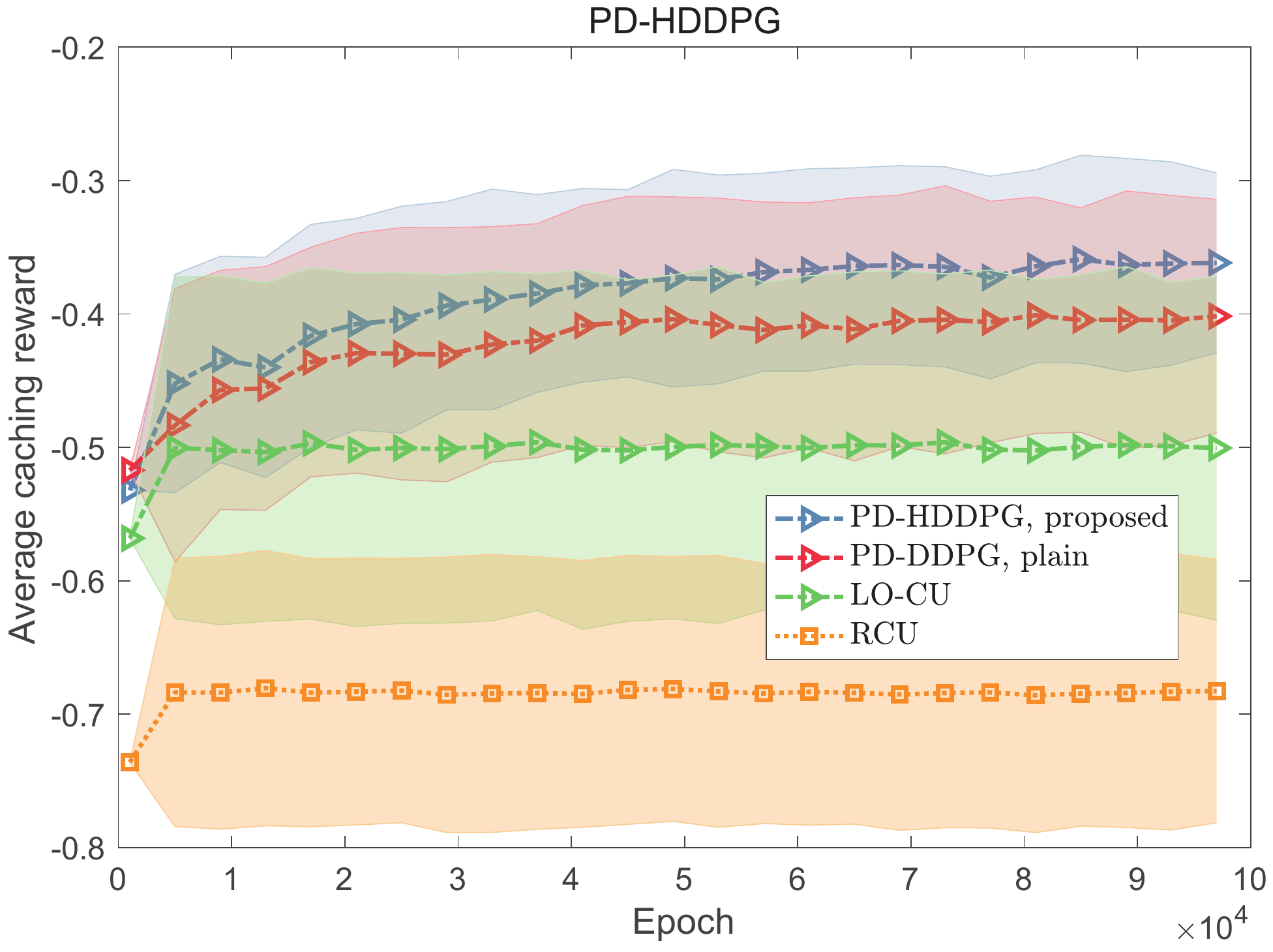}
  \caption{Learning curves of PD-HDDPG. }
  \label{Fig:Convergence-MA-DDPG}
\end{minipage}
\end{figure}

\subsection{Impacts of System Parameters}
In this subsection, we study the impacts of system parameters on the proposed caching framework. All of the results are obtained by averaging over $10^4$ epochs.
We first investigate the impacts of caching capacity under a larger content catalog size (e.g., $F = 50$). 
Clearly, as shown in Fig. \ref{Fig:FCC}, C-HDDPG is always superior to other algorithms. When fractional caching capacity (e.g., $L/ F$) is $10\%$, C-HDDPG achieves the lowest fronthaul traffic loads, e.g., 0.47, in contrast with PD-HDDPG and FD-HDDPG, e.g., 0.52. The superiority of C-HDDPG demonstrates the effectiveness of using global information to enhance SBS collaboration. As fractional caching capacity grows larger, the gap between PD-HDDPG and FD-HDDPG becomes bigger. Indeed, with the aid of a centralized critic to train local policies, PD-HDDPG can allow SBSs to tightly collaborate in comparison to the fully decentralized scheme.
Although FD-HDDPG depends on local observations only, it still outperforms CO-CU and LO-CU by 6.11\% and 11.51\% respectively, under the scenarios being studied. This observation demonstrates the remarkable advantages of using DRL algorithms to learn policies under dynamic environments over conventional optimization-based algorithms. 

\begin{figure}
\begin{minipage}[t]{.495\linewidth}
  \centering
  \includegraphics[scale=0.38]{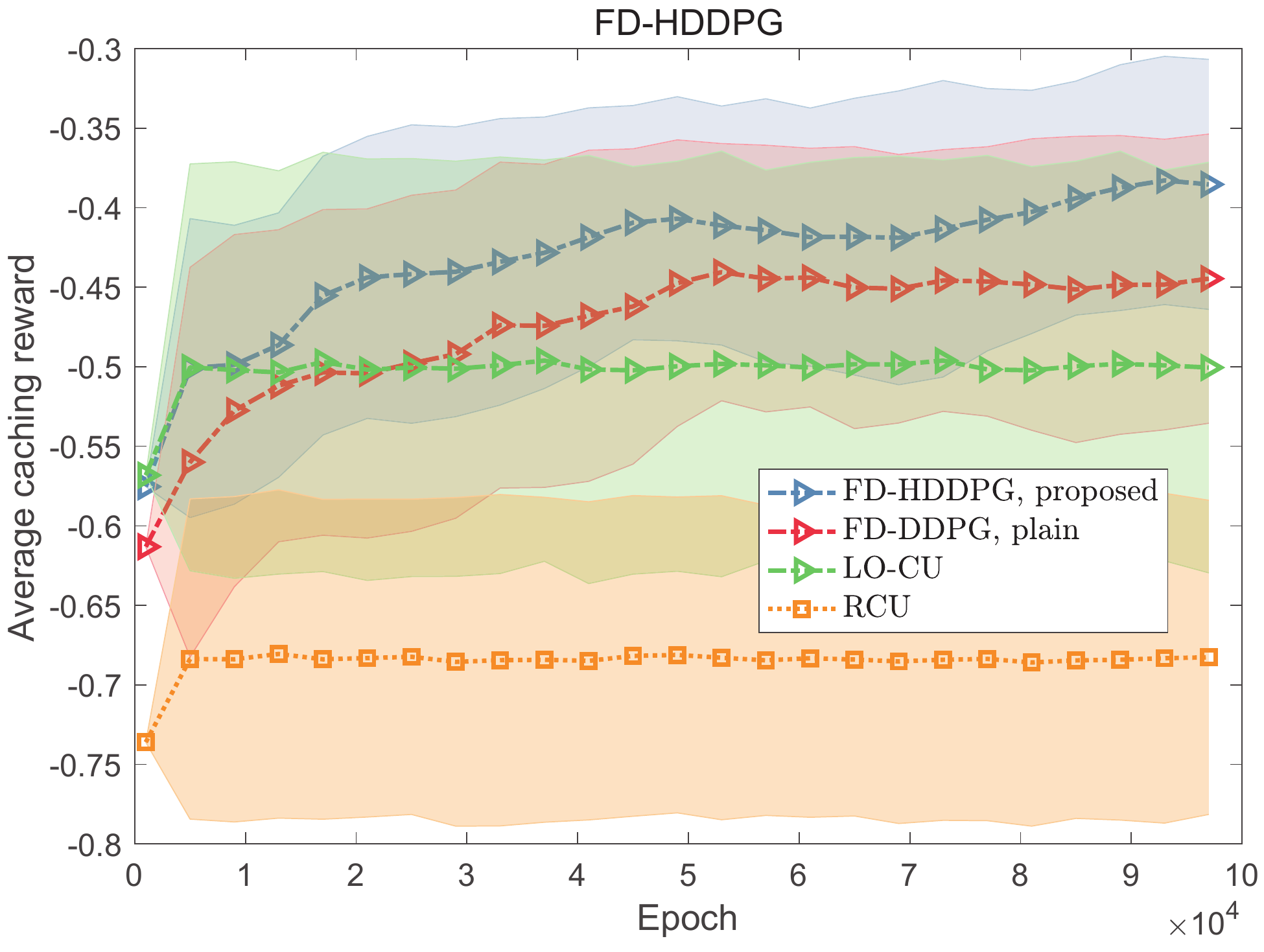}
  \caption{Learning curves of FD-HDDPG. }
  \label{Fig:Convergence-FD-DDPG}
\end{minipage}
\begin{minipage}[t]{.495\linewidth}
  \centering
  \includegraphics[scale=0.5]{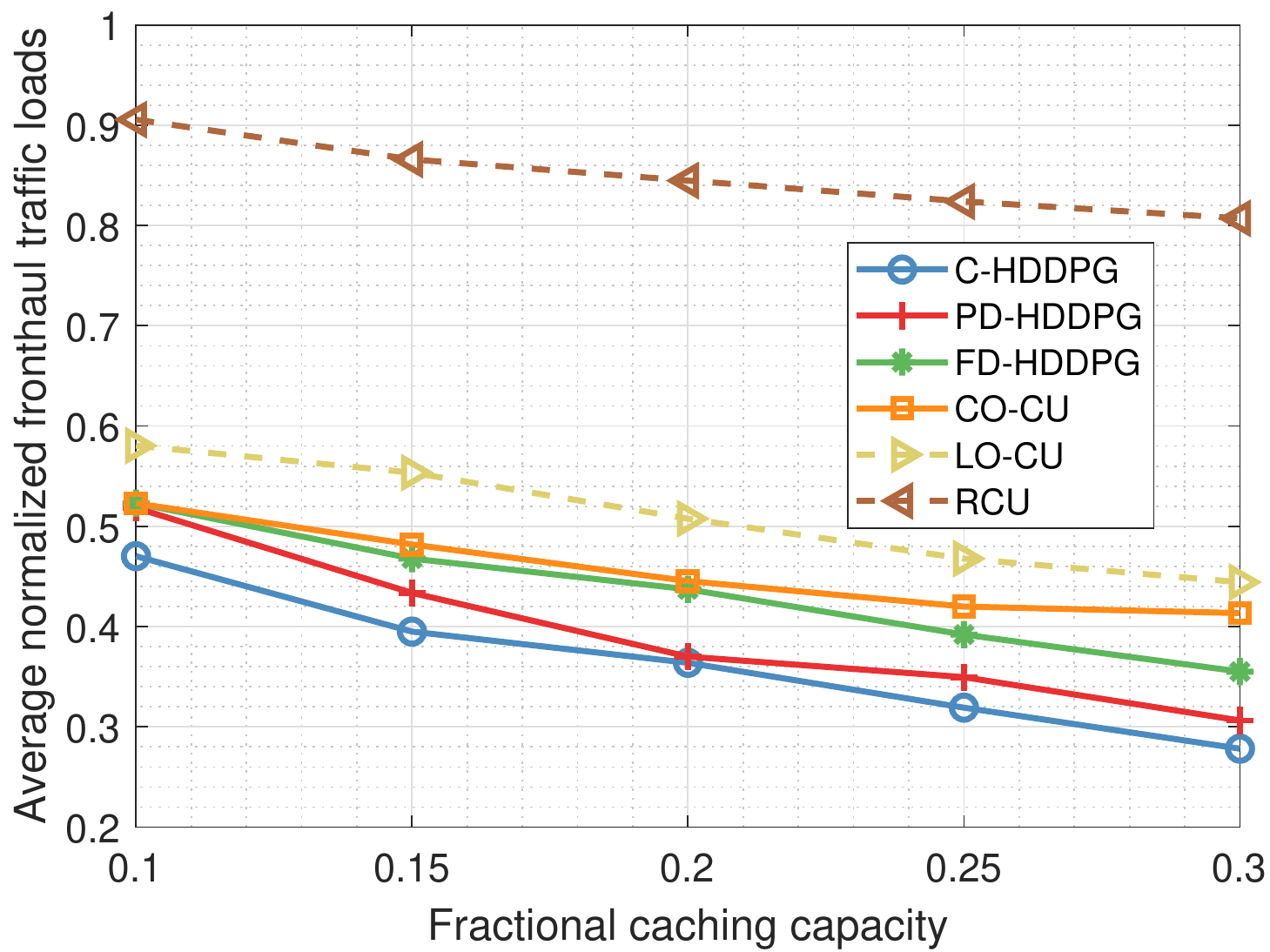}
  \caption{Impacts of fractional caching capacity. }
  \label{Fig:FCC}
\end{minipage}
\end{figure}

Hereunder, we conduct experiments to investigate the impacts of content popularity by varying the skewness factor of Zipf distribution. Moreover, under each scenario being investigated, the corresponding skewness factor is fixed as a constant, in which a larger skewness factor indicates a more concentrated content popularity. 
As can be seen, fronthaul traffic loads decrease as the skewness factor becomes larger for all of the algorithms except for RCU. The reason for this is that user requests are more likely to be accessed in local SBSs if their preferences are more centered. Furthermore, PD-HDDPG achieves comparable fronthaul traffic loads in contrast with C-HDDPG when the skewness factor is smaller than 1; after that, using centralized control only produces a marginal performance gain over decentralized control, yet with a significant implementation cost. This finding demonstrates that PD-HDDPG can efficiently obtain a satisfactory trade-off between complexity and performance.  


To investigate the scalability of the proposed algorithms, we carry out experiments by varying the content catalog size. As depicted in Fig. \ref{Fig:ccs}, PD-HDDPG obtains a comparable performance to C-HDDPG when the content catalog size is smaller than 50; as more content items are considered, C-HDDPG achieves better performance due to utilization of global information.
It is worth noting that, vast gaps can be observed between the proposed algorithms and baselines under either centralized or decentralized scenarios. More specifically, over the entire horizontal axis, C-HDDPG and PD-HDDPG can decrease fronthaul traffic loads by 10.54\% and 7.68\% respectively in comparison to CO-CU; whereas FD-HDDPG can reduce fronthaul traffic loads by 10.40\% compared with LO-CU. All of these results corroborate the scalability of the proposed algorithms. 
Notably, the curve of RCU increases greatly and eventually surpasses 1 as the content catalog size grows larger. This is because the cache updating cost introduced by RCU could overtake fronthaul traffic loads arising from satisfying user requests under large scenarios.
\begin{figure}
\begin{minipage}[t]{.32\linewidth}
  \centering
  \includegraphics[scale=0.33]{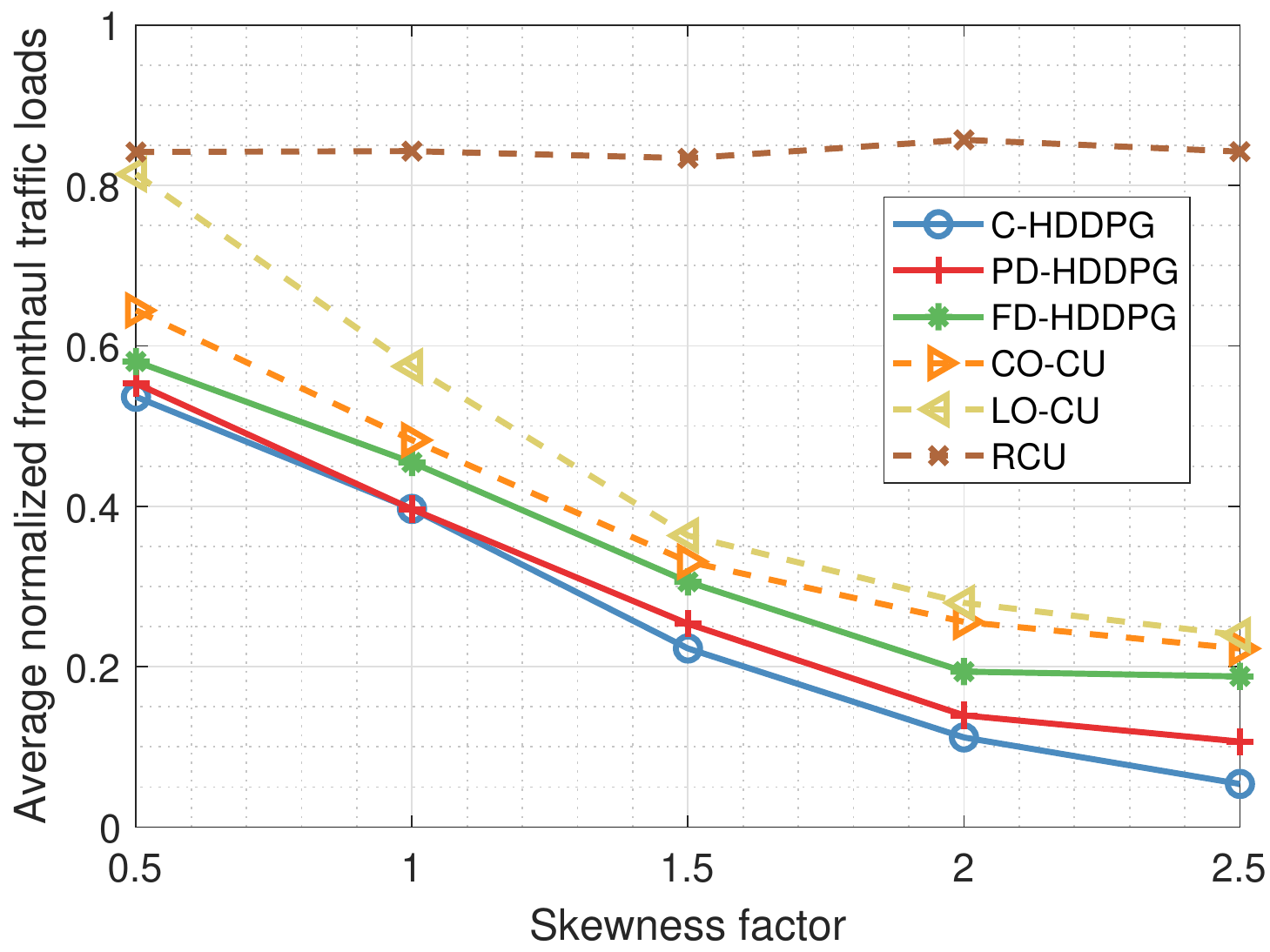}
  \caption{Impacts of content popularity. }
  \label{Fig:SF}
\end{minipage}
\begin{minipage}[t]{.32\linewidth}
  \centering
  \includegraphics[scale=0.33]{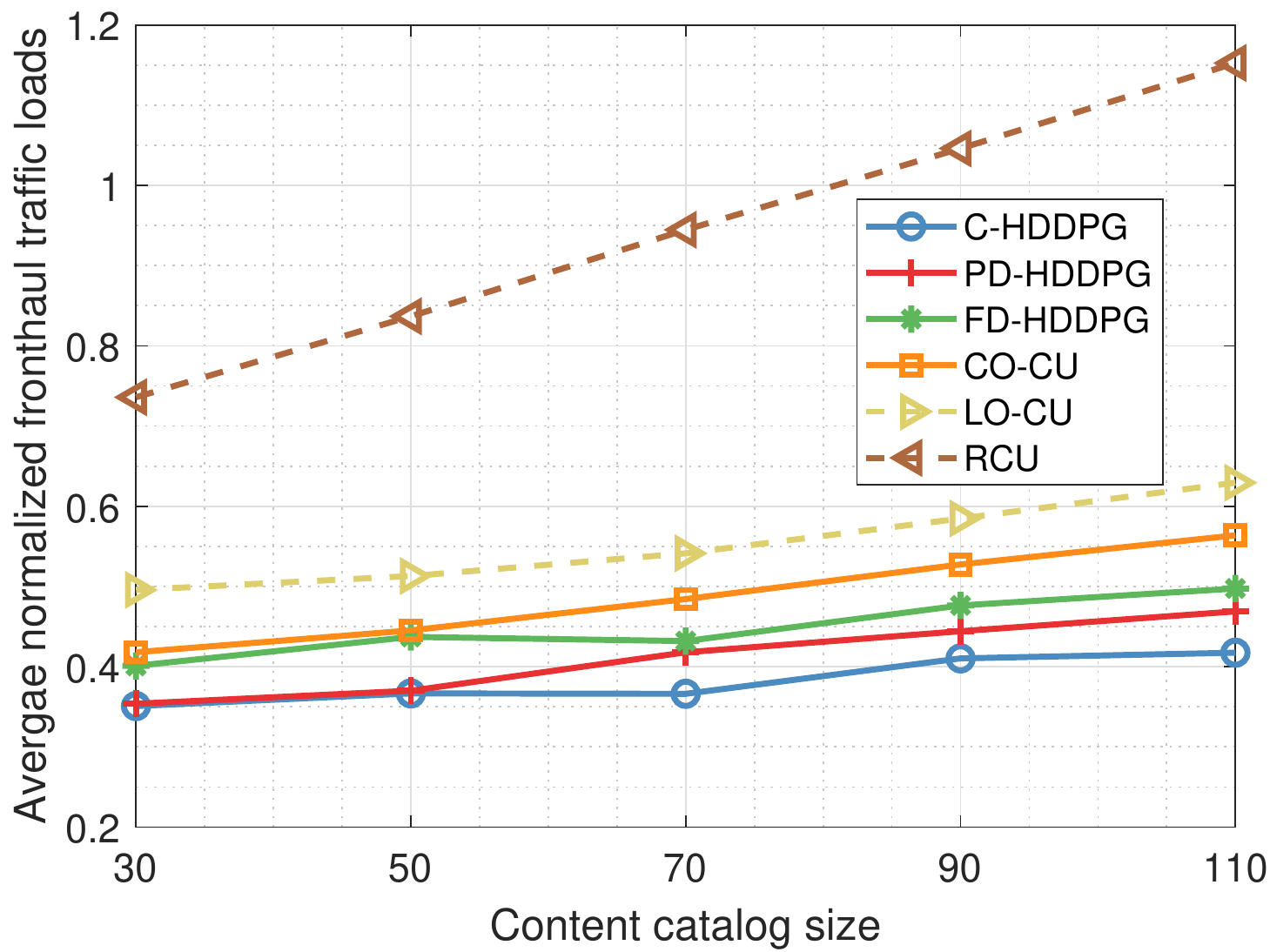}
  \caption{Impacts of content catalog size.}
  \label{Fig:ccs}
\end{minipage}
\begin{minipage}[t]{.32\linewidth}
  \centering
  \includegraphics[scale=0.33]{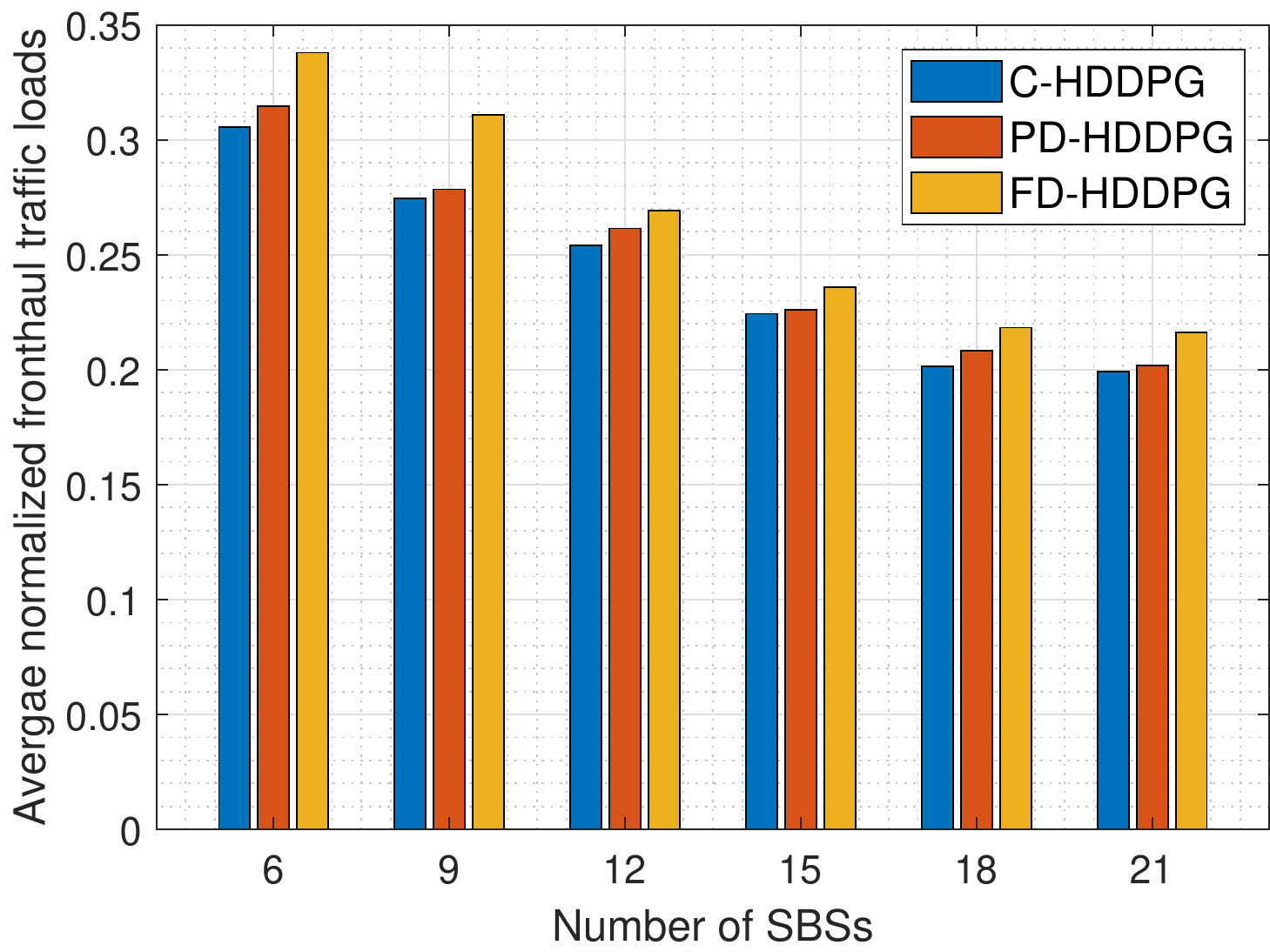}
  \caption{Impacts of number of SBSs.}
  \label{Fig:B}
  \end{minipage}
\end{figure}

We further investigate how the number of agents (i.e., SBSs) impacts the proposed multi-agent algorithms. In these settings, we vary the number of SBSs, and set the distance between two adjacent SBSs as 300 m. As shown in Fig \ref{Fig:B}, traffic loads exhibit a decreasing trend as more SBSs are available to participate in cooperative coded caching. In addition, all of the curves decrease relatively slowly when more than 15 SBSs are deployed, which implies that most users might already be able to access a maximum number of local SBSs, that is usually limited by communication coverage. More importantly, when the number of agents becomes large, PD-HDDPG always achieves comparable results to C-HDDPG  with significant reductions of signaling overhead and complexity. Concerning FD-HDDPG, it achieves slightly larger traffic loads than those of C-HDDPG and PD-HDDPG by at most 3.23\% and 2.34\%, respectively. These results demonstrate the potentials of utilizing decentralized controls as a large number of SBSs are deployed. 
\section{Conclusion}
 We have proposed a deep multi-agent reinforcement learning framework for dynamic cooperative coded caching at small cell networks. Particularly, we have developed a novel deep reinforcement learning algorithm, i.e., homotopy DDPG, to address the challenges arising from the resultant continuous decision-making with constraints. From an engineering perspective, we have proposed centralized, partially decentralized, and fully decentralized controls to balance complexity and performance.  
Simulation results have confirmed that the proposed DRL outperforms plain DDPG under different levels of controls; and the proposed decentralized designs also achieve satisfactory performance compared with the centralized design. 

\appendix
\subsection{\textit{Proof of Proposition \ref{prop:opticon}}}   
\label{appen:A}
Consider an optimal decision sequence $\{\bm A^t\}$, which results in an optimal value $J^* = \sum_{t} (\gamma)^tR^{t+1}$. Suppose that there exists $t', b'$ such that $\sum_{f} a_{f,b'}^{t'} = L'$ and $L' < L$; in addition, the corresponding decision $\bm A^{t'}$ is anticipated to impact rewards $R^{t'+1}$ and $R^{t'+2}$. 
To proceed, we first denote $\mc F' = \{f| a_{f,b'}^{t'} - l_{f,b'}^{t'} > 0\}$. Accordingly, one can create a sequence $\{o_f, \forall f \in \mc F\}$ where $0 \leq o_{f} \leq l_{f,b'}^{t'} - a_{f,b'}^{t'}, \forall f \in \mc F \backslash \mc F',$ and $o_f = 0, \forall f\in \mc F',$ such that $\sum_{f\in \mc F} o_f = L - L'$. Then, we consider a new decision $\underline {\bm A}^{t'} = [\underline a_{f,b}^{t'}]$ for epoch $t'$, where $\underline a_{f,b}^{t'} = a_{f,b}^{t'}$ for $\forall b \neq b'$, and $\underline a_{f,b'}^{t'} = a_{f,b'}^{t'} + o_f$; clearly, this gives rise to $\sum_{f \in \mc F} \underline a_{f,b'}^{t'} = L$. By checking the traffic loads in \eqref{eq:cost}, one can verify that $\underline R^{t'+1} \geq R^{t'+1}$ and $\underline R^{t'+2} \geq R^{t'+2}$. We thereby claim that $\{\underline {\bm A}^t\}$ is also an optimal decision sequence where $\underline {\bm A}^t = {\bm A}^t$ for $\forall t \neq t'$. Hence, Proposition \ref{prop:opticon} holds. 

\subsection{\textit{Proof of Lemma \ref{prop:homgrad}}} 
\label{appen:B}
This proof follows similar procedures to the {\it Deterministic Policy Gradient Theorem} in \cite{silver2014deterministic}. Accordingly, $\nabla J_{\h, \bth} = \nabla_{\bth} \int_{\mc S} p^0 (\bm S) \hq(\bm S, \bm A)|_{\bm A = \siga[\mu_{\bth} (\bm S)]} d\bm S$, where $p^0(\cdot)$ denotes probability density of state at epoch 0.
 Then, following the standard steps in \cite{silver2014deterministic} yields the following 
\begin{align}
 \hspace{-0.5pc} & \nabla_{\bth} \hq (\bm S, \bm A)|_{\bm A = \siga[\mu_{\bth}(\bm S)]} = 
 \nabla_{\bm A} \hq(\bm S, \bm A)|_{\bm A = \siga[\mu_{\bth}(\bm S)]} \nabla_{\bm A'} \siga (\bm A')|_{\bm A' = \mu_{\bth}(\bm S)} \nabla_{\bth} \mu_{\bth} (\bm S) 
  \notag\\ & \qquad \qquad \qquad \qquad \qquad \qquad  \quad ~~~+ 
   \int_{\mc S} \gamma p(\bm S \rightarrow \hat {\bm S}, 1) \nabla_{\bth} \hq (\hat {\bm S}, \hat {\bm A})|_{\hat{\bm A} = \siga[\mu_{\bth}(\hat {\bm S})]} d \hat{\bm S}, \label{eq: proof21}
\end{align}
where $p(\bm S \rightarrow \hat {\bm S}, t_0)$ denotes the probability density of state $\bm S$ transiting to state $\hat{\bm S}$ after $t_0$ epochs; thereafter, one can continue to unfold \eqref{eq: proof21}, resulting in 
\begin{align}
 \hspace{-0.5pc} & \nabla_{\bth} \hq (\bm S, \bm A)|_{\bm A = \siga[\mu_{\bth}(\bm S)]} = \int_{\mc S} \sum_{t = 0}^{\infty} \gamma^{(t)}p (\bm S \rightarrow \hat{\bm S}, t) \nabla_{\bm A} \hq(\hat {\bm S}, \hat{\bm A)}|_{\hat{\bm A} = \siga[\mu_{\bth}(\hat {\bm S})]} \notag\\& ~~~~ \qquad \qquad \qquad \qquad \qquad \qquad  \quad  \cdot \nabla_{\bm A'} \siga (\bm A')|_{\bm A' = \mu_{\bth}(\hat{\bm S})} \nabla_{\bth} \mu_{\bth} (\hat{\bm S}) d \hat{\bm S}.  \label{eq: proof22}
\end{align}
Now, by using \eqref{eq: proof22}, we can obtain 
\begin{align}
  \nabla J_{\h, \bth} = \int_{\mc S} \rho (\bm S)  \nabla_{\bm A} \hq(\bm S, \bm A)|_{\bm A = \siga[\mu_{\bth}({\bm S})]} \nabla_{\bm A'} \siga (\bm A')|_{\bm A' = \mu_{\bth}(\bm S)} \nabla_{\bth} \mu_{\bth} (\bm S) d \bm S,
\end{align}
where $\rho(\bm S)$ denotes the discounted state distribution \cite{silver2014deterministic}. This completes the proof.

\bibliographystyle{IEEEtran}
\bibliography{references}
\end{document}